\documentclass[12pt]{article}

\usepackage[T1]{fontenc}
\usepackage[mathletters]{ucs}
\usepackage[utf8x]{inputenc}
\usepackage[]{geometry}

\usepackage[font={small}]{caption}
\usepackage{amsmath,amsthm, amsfonts}
\usepackage{mlmodern}   %
\usepackage[scaled=1.04,varqu,varl]{inconsolata} %
\usepackage[T1]{fontenc}
\usepackage{bm}

\usepackage[usenames,dvipsnames,svgnames,table]{xcolor}
\PassOptionsToPackage{hyphens}{url}
\usepackage[
  colorlinks=true,
  linkcolor=blue,
  urlcolor=black,
  citecolor=blue,
  backref=page,
  pdfencoding=auto, psdextra
  ]{hyperref}
\usepackage{bookmark}%

\usepackage{cleveref}
\usepackage{braket}

\usepackage[authoryear]{natbib}
\usepackage{graphicx} %
\usepackage{multicol} %
\usepackage{comment,cancel}
\usepackage[shortlabels]{enumitem} %
\usepackage{graphicx}
\usepackage{subcaption}
\usepackage{accents}  %

\usepackage[explicit]{titlesec}

\titleformat{name=\section}[block]{\bfseries\filcenter\thesection}{}{1em}{#1}[]

\titleformat{name=\section,numberless}[block]{\bfseries\filcenter}{}{1em}{#1}[]

\titleformat{name=\subsection}[block]{\filcenter \thesubsection}{}{1em}{#1}[]

\newtheorem{claim}{Claim}
\newtheorem{lemma}{Lemma}
\newtheorem{proposition}{Proposition}
\newtheorem{corollary}{Corollary}

\theoremstyle{definition}
\newtheorem{definition}{Definition}

\usepackage{float}
\usepackage{tikz,xcolor}
\usetikzlibrary{patterns, fit, decorations.pathreplacing, shadings}
\usepackage{tikz-3dplot}

\usepackage{linguex} %

\def\fps@figure{!htbp}
\def\fps@table{!htbp}

\usepackage{easy-todo}

\newcommand*{\ps}{p_2}
\renewcommand{\pm}{p_1}
\newcommand*{\ls}{l_2}
\newcommand*{\lm}{l_1}
\newcommand*{\cs}{c_2}
\newcommand*{\cm}{c_1}

\newcommand*{\an}{a_0}
\newcommand*{\am}{a_1}
\newcommand*{\as}{a_2}

\newcommand*{\sm}{t_1}

\newcommand*{\type}{t}
\newcommand*{\Type}{T}
\newcommand*{\supp}{\operatorname{supp}}

\newcommand*{\prior}{q_0}
\newcommand{\RR}{\mathbb{R}}

\newcommand{\Prior}{\mu_0}

\newcommand{\m}{m}
\newcommand{\M}{M}
\newcommand{\smax}{e_{\mathsf{m}}}

\newcommand{\qc}{\operatorname{qcav}}
\newcommand{\cav}{\operatorname{cav}}
\newcommand{\ev}{\operatorname{ev}}
\newcommand{\eu}{\operatorname{eu}}

\newcommand{\s}{e}
\newcommand{\bq}{\bar{q}}
\newcommand{\qbar}{\bar{q}}
\newcommand{\pbar}{\bar{\pvec}}
\newcommand{\Pbar}{\bar{P}}
\newcommand{\FR}{\text{FR}}
\newcommand{\NR}{\text{NR}}

\newcommand{\pvec}{p}

\newcommand{\qtilde}{\tilde{q}}

\newcommand{\qlbar}{\underline{q}}

\newcommand{\PP}{\mathbb{P}}
\newcommand{\GG}{\mathbb{G}}
\newcommand{\BB}{\mathbb{B}}

\newcommand{\uchi}{\underline{\chi}}
\newcommand{\qcav}{\operatorname{qcav}}
\newcommand{\wcard}{{}\cdot{}} %
\newcommand{\cG}{\mathcal{G}}

\newcommand{\textif}{\text{if}\,\,}
\newcommand{\pc}{\operatorname{pc}}

\makeatletter
\renewcommand\operator@font{\sf}
\makeatother

\title{Credibility in Credence Goods Markets}
\author{
  Xiaoxiao Hu\thanks{Email: \texttt{xhuah@connect.ust.hk.}} 
  \and 
  Haoran Lei\thanks{Hunan University. Email: \texttt{hleiaa@connect.ust.hk.}}
}

\begin{document}
\maketitle
\begin{abstract}
  An expert seller chooses an experiment to influence a client's purchasing decision, but may manipulate the experiment result for personal gain. 
  When credibility surpasses a critical threshold, the expert chooses a fully-revealing experiment and, if possible, manipulates the unfavorable result. In this case, a higher credibility strictly benefits the expert, whereas the client never benefits from the expert’s services. 
  We also discuss policies regarding monitoring expert’s disclosure and price regulation. When prices are imposed exogenously, monitoring disclosure does not affect the client’s highest equilibrium value. A lower price may harm the client when it discourages the expert from disclosing information.

\bigskip

\noindent\textbf{Keywords:} Credence Goods, Credibility, Persuasion, Cheap Talk, Limited Commitment

\noindent\textbf{JEL Code:} D82, D83, L12

\end{abstract}

\clearpage

\thispagestyle{empty}
\tableofcontents
\clearpage

\section{Introduction}
Credence goods are products or services whose qualities are hard for consumers to assess.
Sellers of credence goods often act as experts, advising consumers on the suitability of the goods.
Examples include medical services, automobile repairs, and taxi rides.
Experts in
these markets have a substantial information 
advantage over consumers, and may
exploit it for personal gain.
For example, a doctor may recommend expensive but unnecessary medical treatments, knowing that the patient cannot detect the dishonest behavior.
Such situations abound in real life.%
\footnote{
  In the United States, a doctor was sentenced to prison for administering unnecessary injections and distributing medically unnecessary opioids
  (\url{https://www.justice.gov/opa/pr/doctor-sentenced-role-illegally-distributing-66-million-opioid-pills-and-submitting-250/}).
  A car mechanic from Texas was charged with fraud for using false parts on vehicles and charging customers for uninstalled parts
  (\url{https://abc13.com/richard-sisney-jr-auto-shop-owner-arrested-montgomery-county-fraud-new-caney/5416238/}).
  In the United Kingdom, an investment firm collapsed after the company was revealed to have misled investors about the risks associated with its high-yield bonds
  (\url{https://www.ft.com/content/b881f191-16ac-4eed-84d0-bd6c79d461d7}).%
}

The prevalence of fraud in credence goods markets raises important questions about protecting consumers from exploitation.
One possible solution is to monitor experts' information disclosure.
For example, authorities may require experts to disclose accurate product information,
or establish third-party assessment programs for the quality of expert advice.
However, the effects of such monitoring policies are uncertain, especially when the experts have pricing power.

This paper studies how different levels of monitoring intensity affect the players' welfare.
We examine a basic two-by-two credence goods setup, which involves an expert and a client.
The client is uncertain of the severity of her problem,
and thus can not decide the appropriate treatment.
The expert sets prices for treatments and designs an experiment regarding the problem type.
With some exogenous probability, the expert is monitored and has to truthfully reveal the experiment result.
Otherwise, he can freely manipulate the experiment result.
We refer to the probability of being monitored as the expert's \textit{credibility}.
Additionally, whether the expert is being monitored or not 
is his private information and cannot be observed by the client.

We find that credibility plays a critical role in determining the equilibrium outcome distribution(s).
When credibility is above a critical threshold, there exists a unique equilibrium outcome distribution,
in which the expert designs a fully-revealing experiment and
sets prices that leave no surplus for the client. Moreover, 
the expert's welfare is strictly increasing in his credibility,
but the client never benefits from the expert's services. 
When credibility is below the threshold, 
there exist multiple equilibrium outcome distributions.
Among them, the client's welfare varies, while the expert's welfare remains unchanged and is independent of credibility.
The critical threshold for credibility depends on the distribution of the problem type,
and can be zero when the serious problem is sufficiently likely.

We also examine policies regarding monitoring experts' disclosure and price regulation. 
Concerning protecting clients' welfare,
our findings do not support policies of monitoring experts' disclosure.
As long as the expert controls prices, increasing monitoring intensity (i.e., increasing credibility) generally
does not benefit the client.
Furthermore, if prices are exogenously imposed, the client's maximum equilibrium value
is independent of the monitoring intensity (see \Cref{prp:welfare}).
Among all equilibria with exogenous prices, the equal-margin fully-disclosing equilibria maximize the social welfare.

\subsection{Literature} 
This paper is related to the literature on credence goods, which started from \cite{darbykarni}.
Numerous subsequent works have investigated experts' incentives to defraud clients and
the influences of different institutional characteristics on expert behavior.
Surveys of the extensive literature can be found in \cite{dulleck} and \cite{survey2020jbef}.
To prevent market breakdown caused by information asymmetry, 
existing works either assume that (a)
an expert is liable to fix a client's problem
\citep*{pitchik1987,wolinsky1993,fong2005,liu2011,hu2020}, 
or that (b) expert-provided goods or services are
verifiable by the client \citep{emons1997,dulleck2009,fong2014,bester2018}.
In a comprehensive survey,
\cite{dulleck} term the above two assumptions
as, respectively, the \textit{liability assumption} and the \textit{verifiability assumption}.
This paper imposes the verifiability assumption while excluding the liability assumption.
This is because the client is always allowed to select her most desirable treatment.

As previously mentioned, this paper investigates how the expert's credibility affects equilibrium outcomes and players' welfare.
To the best of our knowledge, this topic has received limited attention in the credence goods literature.
Previous studies have examined related notions such as reputation and trust, mostly within the context of repeated games.
\cite{fong2022trust} investigate the role of trust in situations where the market breaks down in a one-shot game due to inefficiency of the minor treatment.
They show that trust can be sustained in a repeated setting, albeit with some loss of efficiency.
\cite{fong2018reputation} examine the effects of the seller's reputation. 
They show that efficiency is unattainable in a one-shot game but achievable in a repeated context.
\cite{ely2003} study a model featuring heterogeneous experts.
While the good experts' interests are aligned with those of the clients, the bad experts consistently recommend the serious treatment.
They show that good experts' incentives to avoid bad reputations can inadvertently lead to the elimination of profitable interactions and a loss of surplus.

Lastly, there is a growing body of credence goods literature that
allows a flexible transmission protocol between the expert and the client.
\cite{hu2022} model the expert's recommendation as a cheap-talk message,
allowing the client to freely choose any treatment after receiving the expert's recommendation. 
In contrast, this paper assumes that the expert's message may come from a credible source with a given probability.
As discussed later, the binary version of \cite{hu2022} can be viewed as an extreme case of our model.
\cite{li2023} assume divisible credence goods and that the expert has full commitment to his chosen experiment.
They characterize the expert's optimal unit price and information structure, 
and discuss when the client consumes beyond her needs.
\cite{emons2022} assume that the expert only observes the most suitable treatment for the client, and model the expert's recommendation as a subset of available treatments.

\section{Model} \label{sec:model}

There are two players: Client (she) and Expert (he).
Client has a problem, which is either minor ($t_1$) or
serious ($t_2$). 
Client is uncertain of the \textit{problem type} $t$,
whose distribution follows:
\[ 
  \Prior (t_2) = \prior \text{ and } \Prior (\sm) = 1-\prior.
\]
Expert is able to provide two kinds of treatments: the minor treatment ($\am$) and the serious treatment ($\as$).
The two treatments differ in efficacy: the minor treatment cures the minor problem only whereas the serious treatment cures both the minor and the serious problems.
Denote the costs of minor and serious treatments as $\cm$ and $\cs$, respectively.

\paragraph{Information.}
Before Client's consultation, Expert publicly chooses an experiment, generating a message $m$ whose distribution
is contingent on the problem type.
Let $\Type \equiv \{\sm, t_2\}$ and an experiment is specified by
$$
ξ: T \to  \Delta M, \text{ where $M$ is the message space.}
$$
Assume that the generated message is privately observed by Expert.
With some \textit{exogenous} probability $\chi \in (0,1)$, Expert is
\textit{credible} and truthfully reports the generated message to Client;
with the complementary probability $1 - \chi$,
Expert is \textit{non-credible} and freely chooses which message to send.%
\footnote{There are two interpretations of the limited credibility setting.
    First, as mentioned in the introduction, the probability $\chi$ may reflect the
    monitoring intensity. For example, $\chi$ can be the
    probability that Expert will be monitored by some third party.
    Second, \cite{lipnowski2022} interpret the expert's limited credibility as the result of
    ``weak institution.'' Specific to our credence-good setting, consider a diagnostic clinic (Expert) that devises tests to determine the patient's problem type. If the clinic is not independent, it may be susceptible to the influence of pharmaceutical vendors. In such cases, the clinic could potentially distort the diagnostic outcome to benefit the vendors.
}
Assume that Expert privately observes his credibility type, and that he observes it after choosing the experiment $ξ$.

Upon receiving the message, Client chooses $a \in A \equiv \{\an,\am, \as\}$
where $\an$ means purchasing no treatment.
The ``verifiability assumption'' of credence goods 
(i.e., the performed treatment is observable and verifiable by Client) is imposed.
Timing of the game is as follows:

\begin{enumerate}
\item
  Nature determines Client's problem type $\type \in \{\sm, t_2\}$ according to the distribution $\Prior$.
\item
  Expert posts a price list $\pvec = (\pm, \ps) \in \RR_+^2$, where $\pm$ and $\ps$ are the prices for the minor and the serious treatments respectively. 
  Simultaneously, Expert publicly announces an experiment $ξ: \Type \to  \Delta \M$. 
\item 
  Expert privately observes the realized problem type%
      \footnote{We assume that Expert observes the problem type, allowing the sent message to 
      depend on it whenever Expert can manipulate the experiment result. However, in reality, Expert may only observe the experiment result and cannot observe the problem type. We investigate this possibility in \Cref{sec:discussion-type}. %
      }
  and the message generated by the experiment.   
  With probability $\chi$, the generated message is truthfully disclosed to Client.
  Otherwise, Expert strategically sends some message $m \in \M$ to Client.
\item 
  Upon receiving the message, Client
  chooses an action $a \in A = \{\an, \am, \as\}$.
\item 
  Client and Expert obtain their payoffs, $u^C(t, a, \pvec)$ and $u^E(a, \pvec)$, respectively:
  \begin{equation*}
   u^C(t_i, a_j, \pvec) = -p_j -\mathbf{1}(j < i) \cdot l_i \, ,
  \quad
   u^E ( a_j, \pvec) = p_j - c_j
  \end{equation*}
  where $\mathbf{1}(j < i)$ is the indicator function that takes the value of $1$ when $j < i$ and $0$ otherwise.

  \Cref{tab:payoff} shows the ex post payoffs for both players, where each column corresponds to some Client's final decision and each row corresponds to a specific problem type.
  In each table cell, the left term represents Client's payoff and the right term represents Expert's payoff.
\end{enumerate} 

\begin{table}[H]
\centering
\captionsetup{justification=centering,margin=1.6cm}
\begin{tabular}{llll}
  &\quad $\an$ &\quad\quad\quad $\am$  &\quad\quad $\as$ \\ 
  \cline{2-4} 
  \multicolumn{1}{l|}{$\sm$} & \multicolumn{1}{c|}{$-\lm,0$} & 
  \multicolumn{1}{c|}{$ - \pm, \pm - \cm$}  & \multicolumn{1}{c|}{$-\ps, \ps - \cs$} \\ \cline{2-4} 
  \multicolumn{1}{l|}{$t_2$} & \multicolumn{1}{c|}{$-\ls,0$} & 
  \multicolumn{1}{c|}{$-\pm - \ls,\pm - \cm$} & \multicolumn{1}{c|}{$-\ps, \ps - \cs$}  \\ \cline{2-4} 
  \end{tabular}
  \caption{Client's and Expert's ex post payoffs}
\label{tab:payoff}
\end{table}

\paragraph{Assumptions} Throughout this paper, we impose the following assumptions:
(i) the message space $M$ is finite and sufficiently large;%
    \footnote{Since $|T|=2$ and $|A|=3$, it suffices to impose that $|M| \ge 3$.}
(ii) the serious treatment is more costly than the minor treatment: $\cs > \cm$;
(iii) both treatments are efficient: $\ls > \cs$ and $\lm > \cm$;  
(iv) the serious treatment implemented on the serious problem has a higher surplus than that of 
the minor treatment implemented on the minor problem: $\ls - \cs > \lm - \cm$.
Note that assumptions (ii) and (iv) imply $l_2 > l_1$.

\subsection{Solution concept}
Our solution concept is the \textit{Expert-preferred perfect Bayesian equilibrium} (henceforth, \textit{equilibrium}).
We denote the entire game as $\cG$ and the subgame following Expert posting price list $\pvec$ as $\cG_{\pvec}$.
Refer to a perfect Bayesian equilibrium of $\cG_{\pvec}$ as a \textit{$\pvec$-equilibrium}, defined as below.

\begin{definition}\label{def:p-eq}
  A \textit{$\pvec$-equilibrium} consists of four maps:
  the experiment $ξ: T \to Δ M$;
  Expert's signalling strategy $σ: T \to Δ M$;
  Client's strategy $ρ : M \to Δ A$;
  and Client's belief updating
  $η : M \to Δ T$; such that 
  \begin{enumerate}
  \item 
    Given Expert's chosen experiment $ξ$ and signalling strategy $σ$, Client's posterior after receiving $m \in \M$ is
    obtained using Bayes's rule:
    $$
    \eta(\type \mid \m) = \frac
        {\Prior(t) [\chi ξ(m \mid t) + (1 - \chi)σ(m \mid t)]}
        {\sum_{t' \in T} \Prior (t') [\chi ξ(m \mid t') + (1 - \chi)σ(m \mid t') ] }\quad \text{ for all $t \in T$;}
    $$
  
  \item
    Given the belief updating $\eta$, Client's strategy
    $\rho (\wcard \mid \m)$ after receiving message $m$ is supported on 
    $$
    \arg\max_{a \in A} \sum_{\type \in \Type} 
      \eta(\type \mid \m) u^C(a, \type , \pvec) \, ;
    $$ 
    
  \item 
    Given Client's strategy $\rho$, Expert's signalling strategy
    $σ (\wcard \mid \type)$ is supported on 
    $$
    \arg\max_{m \in \M} \sum_{a \in A} \rho(a \mid \m) u^E (a, \pvec) \quad \text{for all $\type \in \Type$.}
    $$ 
  \end{enumerate}    
\end{definition}

When there exist multiple $\pvec$-equilibria for some price list $\pvec$, we focus on the \textit{Expert-optimal $\pvec$-equilibria} that yield the highest payoffs for Expert among all $\pvec$-equilibria.
Denote by $(ξ, \eta, \rho, σ)$ an Expert-optimal $\pvec$-equilibrium and then Expert's ex ante payoff is:
\[
v^*_{\chi} (\Prior ; \pvec) = 
  \sum_{t \in T} \Prior(t) \Big( 
    \sum_{a \in A} \rho (a \mid {\hat{m}(t)}) u^E(a, \pvec)
  \Big)
\text{ where $\hat{m}(t) \in \supp σ(\wcard \mid t)$.}  
\]
Refer to $v^*_{χ} (\Prior ; \pvec)$ as Expert's \textit{$\pvec$-equilibrium value}.

Finally, Expert chooses some price list to maximize his ex ante payoff of the entire game.
Without loss,%
    \footnote{When  $p_i > l_i$ for some $i \in \set{1,2}$, 
    Client will not purchase treatment $a_i$ at any posterior, and decreasing such $p_i$ to $l_i$ will not affect the outcome distribution.
    If $p_i < c_i$ for some $i \in \set{1,2}$, increasing such $p_i$ to $c_i$ will weakly increase Expert's ex ante payoff.
    Finally, when $p_2 < p_1$, Client will not purchase treatment $a_1$ at any posterior, and decreasing such $p_1$ to $p_2$ will not affect the outcome distribution.
    }
we restrict attention to the set of price lists
$P \equiv \{(p_1, p_2) \colon p_1 \in [c_1, l_1], p_2 \in [c_2, l_2] , p_2 \ge p_1 \}$.
Expert's maximization problem is:
$$
\ev^*_{χ} (\Prior) = \max_{\pvec \in P} v^*_{\chi} (\Prior ; \pvec),
$$
where $\ev^*_{χ} (\Prior)$ is referred to as Expert's \textit{equilibrium value}.

\section{Benchmarks}
We first analyze the two extreme cases: no credibility ($\chi = 0$) and
full credibility ($\chi = 1$).
As the problem type is binary, we slightly abuse notations and
denote Expert's $\pvec$-equilibrium value and equilibrium value by, respectively, $v^*_{χ} (\prior ; \pvec)$ and $\ev^*_{χ} (\prior)$, where $\prior = \Prior (t_2)$.

\subsection{No credibility}
When $\chi = 0$, our model reduces to the scenario in which Expert first sets prices and then employs cheap talk \citep{crawford1982}
to influence Client's purchasing decision.
This situation has been analyzed in \cite{hu2022}.
Below we briefly summarize their main results 
and introduce notations for later usage.

Fixing some price list $\pvec \in P$,
let $a^* (q; \pvec)$ be the set of Client's (possibly mixed)
optimal choices at the belief $q$. %
The corresponding set of Expert's payoffs is $V(q ; \pvec) \equiv \set{u^E(\tilde a , \pvec) \mid \tilde a \in a^* (q ; \pvec) }$.
Let $v (q ; \pvec) \equiv \max V(q ; \pvec)$ and refer to $v (q ; \pvec)$ as Expert's belief-based \textit{indirect utility}.
By Berge's maximum theorem,
$v (\wcard ; \pvec) : [0,1] \to \RR$ is an upper semicontinuous function.
Denote by  $\qc v (\wcard ; \pvec)$ the \textit{quasiconcavification} of $v (\wcard ; \pvec)$;
that is, $\qc v (\wcard ; \pvec)$ is the lowest quasiconcave function such that $\qc v (q' ; \pvec) \ge v (q' ; \pvec)$ for all
$q' \in [0,1]$. %
\cite{lipnowski2020} show that Expert's $\pvec$-equilibrium value given $χ = 0$ is exactly the quasiconcavification of his
belief-based indirect utility:
$v^*_{0} (\prior ; \pvec) = \qc v (\prior ; \pvec)$.
Let $π (\prior) \equiv  \max_{\pvec \in P}v(\prior ; \pvec)$ be Expert's profits under discriminatory pricing.
\cite{hu2022} further characterize Expert's equilibrium value as the quasiconcavification of $π$:
$\ev^*_{0} (\prior) = \qc π(\prior)$. That is,
\begin{align*}
\ev^*_{0} (\prior ) = 
	\begin{cases}
		\lm - \cm & \text{ if } \prior \in \big[0, \frac{\cs - \cm}{\ls - \lm}\big]\\
		\prior (\ls - \lm) + (\lm - \cs) & \text{ otherwise.}  
	\end{cases}
\end{align*}

\subsection{Full credibility}

When $\chi = 1$, our model reduces to the scenario in which Expert first sets prices and then employs persuasion \citep{kamenica2011}
to influence Client's purchasing decision.
Then there exists an equilibrium in which Expert sets $\ps = \ls$, 
$\pm = \lm$ and designs a fully-revealing experiment.
It can be verified that this equilibrium allows Expert to extract the full surplus.
Denote by $\cav π( \wcard)$ the \textit{concavification} of $π(\wcard)$;
that is, $\cav π(\wcard)$ is the lowest concave function that is everywhere above $π(\wcard)$.
Then Expert's equilibrium value is
$$
\ev^*_{1} (\prior ) = \prior(l_2 - c_2) + (1 - \prior)(l_1 - c_1) = \cav π(\prior).
$$

To sum up, Expert's equilibrium values with no credibility and with full credibility are $\qc π(\prior)$ and $\cav π(\prior)$, respectively.
\Cref{fig:profit} illustrates the two extreme cases:
the dashed line and the solid curve represents Expert's equilibrium value with full credibility and with no credibility, respectively.

\begin{figure}[H]
  \centering
  \includegraphics[]{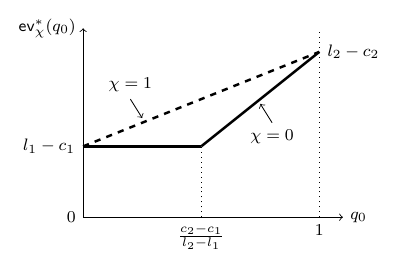}
  \caption{Expert's equilibrium values when $χ=0$ and when $χ=1$}
  \label{fig:profit}
\end{figure}

\section{Limited credibility}

We turn to our main model for $\chi \in (0,1)$, meaning that Client is uncertain
of whether Expert is credible or non-credible.
Divide the space of price lists $P$ into three regions:
$P^1 =  \set{(p_1, p_2) \in P \colon p_1 - c_1 > p_2 - c_2}$,
$P^2 =  \set{(p_1, p_2) \in P \colon p_2 - c_2 > p_1 - c_1}$ and 
$\Pbar = \{ (p_1, p_2) \in P \colon p_1 - c_1  = p_2 - c_2\}$.
We first argue that it is never optimal for Expert to choose $\pvec$ from $P^1$.
Specifically, for any  $\pvec \in P^1$,
Expert's $\pvec$-equilibrium value must be strictly lower than his 
$\pvec'$-equilibrium value
where $\pvec' = (l_1, l_1 - c_1 + c_2)$.

\begin{lemma}  \label{lmm:1}
  Fixing any $\pvec \in P^1$, we have
  $v_χ^* (\prior ; \pvec) < v_χ^* (\prior ; l_1, l_1 - c_1 + c_2)$ 
  for all $\prior \in (0,1)$ and $χ \in (0,1)$.
\end{lemma}

\begin{proof}
  See \Cref{app:lmm-1}.
\end{proof}

\Cref{lmm:1} allows us to focus on price lists from $\Pbar \cup P^2$
for calculating Expert's equilibrium value.
In this section, we characterize $v^*_{\chi}(\prior; \pvec)$ for $\pvec \in \Pbar \cup P^2$,
and then obtain Expert's equilibrium value by calculating the upper envelope of
$v^*_{\chi}(\prior; \pvec)$ as $\pvec$ varies.

\subsection{Expert's $\pvec$-equilibrium value}
\label{sec:p-equil}
Fixing some $\pvec \in \Pbar \cup P^2$, we establish the upper and lower bounds for Expert's $\pvec$-equilibrium value. 
On the one hand, by using the same experiment as the signaling strategy (i.e., $ξ = σ$), 
Expert can achieve the payoff when he uses cheap talk to influence Client's decision. 
Since Expert's highest payoff under cheap-talk communication
is $\qc v (\prior ; p )$ \citep{lipnowski2020}, 
it follows that
$v^*_{\chi}(\prior ; \pvec) \ge \qc v (\prior;\pvec)$ for all $\prior \in (0,1)$ and $χ \in (0,1)$.
On the other hand, Expert's highest payoff under persuasion is
$\cav v (\prior ; \pvec)$ \citep{kamenica2011}, which serves 
as an upper bound for Expert's $\pvec$-equilibrium value.
The concave and quasiconcave envelopes can be written as:
\begin{equation} %
  \notag
  \qc v (\prior ; \pvec) = \begin{cases}
    p_1 - c_1 & \text{ if } \prior < \qbar(\pvec);\\
    p_2 - c_2 & \text{ if } \prior \ge \qbar(\pvec),
  \end{cases}    
  \,
  \cav v (\prior ; \pvec) = \begin{cases}
    \lambda(\pvec) \prior  + p_1 - c_1 & \text{ if } \prior < \qbar(\pvec);  \\
    p_2 - c_2 & \text{ if } \prior \ge \qbar(\pvec),
  \end{cases}    
\end{equation}\normalsize
where the cutoff $\qbar(\pvec)$ is
\begin{equation}\label{eq:qbar}
\qbar(\pvec) = \begin{cases}
 \frac{\ps - \lm}{\ls - \lm} & \text{if } p_2 \ge l_2 - \frac{l_2 - l_1}{l_1} p_1 \\
  \frac{\ps - \pm}{l_2} & \text{otherwise}
\end{cases}
\end{equation}
and $\lambda(\pvec) = \frac{p_2 - c_2 - (p_1 - c_1)}{\qbar(\pvec)}$.
\Cref{fig:xiong666} illustrates the possible cases of
$v (\prior ; \pvec)$, $\qc v (\prior ; \pvec)$ and $\cav v (\prior ; \pvec)$,
depending on the functional form of $\qbar(\pvec)$.

\begin{figure}[!htbp]
  \begin{changemargin}{-2.5cm}{-2.5cm}
    \begin{center}
      \includegraphics[width=0.351\textwidth]{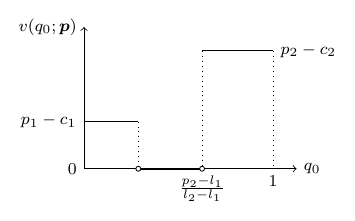}
      \includegraphics[width=0.351\textwidth]{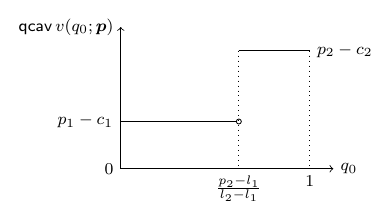}
      \includegraphics[width=0.351\textwidth]{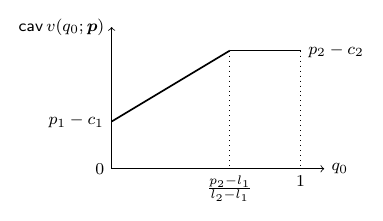}\\
      {\footnotesize (a) Case I: $p_2 \ge l_2 - \frac{l_2 - l_1}{l_1} p_1$}  
      
      \includegraphics[width=0.351\textwidth]{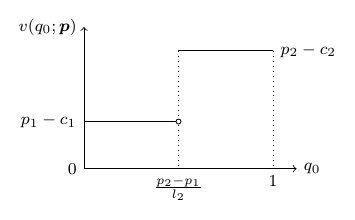}
      \includegraphics[width=0.351\textwidth]{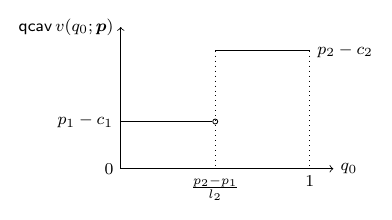}
      \includegraphics[width=0.351\textwidth]{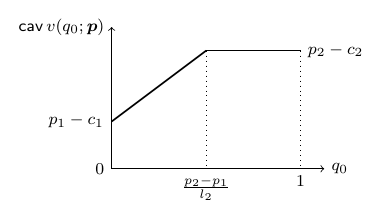}\\
      {\footnotesize (b) Case II: $p_2 < l_2 - \frac{l_2 - l_1}{l_1} p_1$}          
    \end{center} 
  \end{changemargin}
  \caption{Possible cases of $v (\prior ; \pvec)$, $\qc v (\prior ; \pvec)$ and $\cav v (\prior ; \pvec)$ \label{fig:xiong666}}
\end{figure}

When both treatments yield the same margin for Expert (i.e., $\pvec \in \Pbar$), we have
$\cav v (\prior ; \pvec) = \qc v (\prior ; \pvec) = p_2 - c_2$ for all $\prior \in [0,1]$.
Since the upper bound and the lower bound coincide,
$v^*_{\chi} (\prior ; \pvec) = p_2 - c_2$ for all $χ \in (0,1)$.
Furthermore, when the serious problem is sufficiently likely such that $\prior \ge \qbar(\pvec)$, we have
$\cav v (\prior ; \pvec) = \qc v (\prior ; \pvec) = p_2 - c_2$ and thus
$v^*_{\chi} (\prior ; \pvec) = p_2 - c_2$ for all $χ \in (0,1)$.
Below we characterize Expert's $\pvec$-equilibrium value when
$\pvec \in P^2$ and $\prior < \qbar(\pvec)$.
It turns out that then $v^*_{\chi}(\prior; \pvec)$ will be either $\qc v (\prior;\pvec)$ or $\cav v (\prior ; \pvec)$,
depending on whether credibility is above the threshold
$$
\uchi ( \prior ; \pvec) ≡ \frac{\qbar(\pvec) - \prior}{ \qbar(\pvec) (1 - \prior)}.
$$

\begin{proposition}
  \label{prp:p-eq-value}  
  If $\pvec \in P^2$ and $\prior < \qbar(\pvec)$, then Expert's $\pvec$-equilibrium value is:
  \begin{equation*}
    v^*_{\chi}(\prior; \pvec) = \begin{cases}
      \qc  v (\prior ;\pvec)  & \text{ if } χ < \uchi ( \prior ; \pvec); \\
      \cav v (\prior ;\pvec)  & \text{ if } χ \ge \uchi ( \prior ; \pvec).
    \end{cases}  
  \end{equation*}
\end{proposition}

\begin{proof}
  See \Cref{app:prp-p-eq-value}.
\end{proof}

\Cref{prp:p-eq-value} implies a \textit{switch of payoff mode} as Expert's credibility changes. 
When $\chi \ge \uchi ( \prior ; \pvec)$, Expert's $\pvec$-equilibrium value is $\cav v (\prior ; \pvec)$,
referred to as the \textit{persuasion mode of payoff.} 
When $\chi < \uchi ( \prior ; \pvec)$, 
Expert's $\pvec$-equilibrium value is $\qcav v (\prior ; \pvec)$,
referred to as the \textit{cheap-talk mode of payoff.} 
\Cref{fig:mode-switch} shows Expert's $\pvec$-equilibrium value 
for $χ \in [0,1]$ and $\prior \in [0, \bq(\pvec)]$.
The boundary between the two payoff modes,
$\uchi ( \prior ; \pvec)$,
is denoted by the dashed curve on the $\prior$--$\chi$ plane.
Note that $\uchi ( \prior ; \pvec)$ is strictly decreasing in $\prior$ over the interval $[0 , \qbar(\pvec) ]$,
with $\uchi ( 0 ; p ) = 1$ and $\uchi ( \qbar(\pvec) ; p ) = 0$ at the two endpoints.
This can be intuitively understood as follows. Using cheap-talk, Expert can only obtain
the cheap-talk mode of payoff $\qcav v (\prior ; \pvec) = p_1 - c_1$.
To achieve the persuasion mode of payoff,
Expert must make Client more likely to purchase the 
higher-margin treatment $a_2$ than she would in the cheap-talk scenario.
A lower $\prior$ indicates that Client is less willing to choose $\as$, demanding a higher 
credibility to facilitate the persuasion.

\begin{figure}
  \centering
  \begin{subfigure}{0.49\textwidth}
  \centering
  \includegraphics[width=\textwidth]{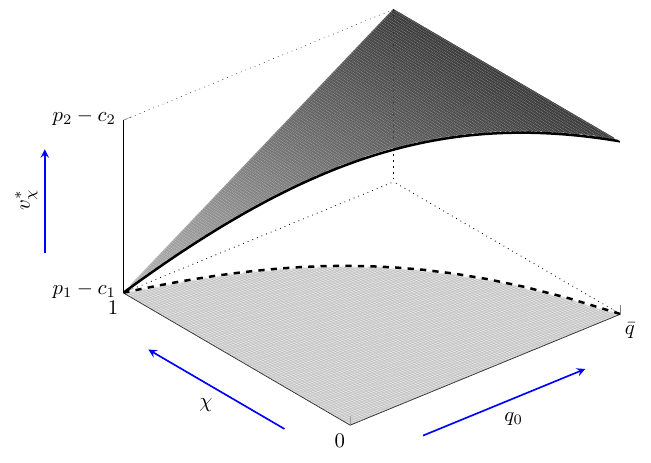}
  \caption{\footnotesize Illustration of $v^*_χ (\prior ; \pvec)$ \label{fig:mode-switch}}
  \end{subfigure}
  \hfill
  \begin{subfigure}{0.49\textwidth}
    \centering
      \includegraphics[width=\textwidth]{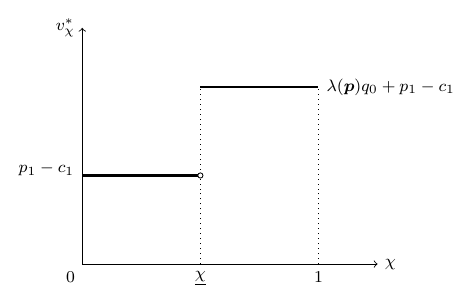}
      \caption{\footnotesize Illustration of $v^*_χ (\prior ; \pvec)$ for a fixed $\prior < \qbar (\pvec)$ \label{fig:trust}}
  \end{subfigure}
  \caption{Expert's $\pvec$-equilibrium value}
\end{figure}

\Cref{fig:mode-switch} also suggests discontinuous jumps of $v_χ^*$ at the dashed curve $\uchi ( \prior ; \pvec)$.
To make it more explicit how credibility affects Expert's $\pvec$-equilibrium value, 
we illustrate $v^*_{\chi} (\prior ; \pvec)$ %
for some fixed $\prior < \qbar (\pvec)$ in
\Cref{fig:trust}, where there is a jump at $\uchi (\prior ; \pvec)$.%
    \footnote{\cite{lipnowski2022} refer to the discontinuity at $χ = \uchi (\prior ; \pvec)$ 
    as \textit{collapse of trust};
    that is, an arbitrarily small reduction in  credibility can cause a large drop in Expert's $\pvec$-equilibrium value.
    They establish the existence of collapse of trust in a more general setting 
    (Lipnowski et al., 2022, Proposition 2, page 2725).}
Note that Expert's $\pvec$-equilibrium value is independent of $χ$ over the interval $[\uchi ( \prior ; \pvec) , 1]$.
One may wonder why Expert will not benefit from a higher credibility given $χ \ge \uchi (\prior ; \pvec)$,
and the reason is as follows. 
When Expert
is fully credible, he chooses a partially-revealing experiment 
that induces the posterior of either $q = 0$ or $q = \qbar(\pvec)$.
As Expert becomes less credible,
he will choose a more informative experiment as a compromise, still 
achieving the payoff of $\cav v (\prior ; \pvec)$
(See \Cref{prp:experiment} for further discussions on Expert's disclousure). Specifically,
at the cutoff $χ = \uchi(\prior; \pvec)$, Expert will choose a fully-revealing experiment.

To sum up, Expert's $\pvec$-equilibrium value for $\pvec \in \Pbar \cup P^2$ can be written as
\begin{equation}\label{eq:p-eq-value}
  v^*_{\chi}(\prior; \pvec) = \begin{cases}
    p_2 - c_2 & 
      \textif \prior > \qbar (\pvec) \\
    \lambda(\pvec) \prior + p_1 - c_1  & 
      \textif \prior \in [\qlbar (χ ; \pvec) , \qbar(\pvec) ]   \\
    p_1 - c_1  & 
      \textif  \prior < \qlbar (χ ; \pvec)
  \end{cases}  
\end{equation}
where $\qlbar (χ ; \pvec)$ is the implicit function determined by
$\uchi (\qlbar ;\pvec) = χ$.%
  \footnote{It can be verified that $\qlbar (χ ; \pvec)$ is strictly decreasing in $\chi$,  
  with $\qlbar (1 ; \pvec) = 0$ and $\qlbar (0 ; \pvec) = \qbar (\pvec)$ at the endpoints.
  }
Specifically, when $\pvec \in \Pbar$,
we have $\lambda (\pvec) = 0$ and the expression \eqref{eq:p-eq-value} reduces to
$v^*_{\chi}(\prior; \pvec) = p_2 - c_2$ for all $χ \in [0,1]$ and $\prior \in (0 , 1)$.

\subsection{Expert's disclosure in $\pvec$-equilibria}

We turn to characterize Expert's chosen experiment and signaling strategy in Expert-optimal $\pvec$-equilibria.
To introduce some notations,
refer to $(q, \s) \in [0,1] \times \RR$ as an \textit{outcome of $\cG_{\pvec}$},
where $q$ is Client's posterior and $\s$ is Expert's ex ante payoff.  
Denote by $\PP \in \Delta ([0,1] \times \RR)$ the ex ante joint distribution of Client's posterior and Expert's ex ante payoff in the subgame $\cG_{\pvec}$.
We refer to $\PP$ as an \textit{outcome distribution of $\cG_{\pvec}$}.
Furthermore, an outcome distribution of $\cG_{\pvec}$ is called a \textit{$\pvec$-equilibrium outcome distribution} if it is induced by some $\pvec$-equilibrium.%
  \footnote{Let $\supp \PP \equiv \set{(q_1, e_1) \dots, (q_K, e_K)}$ for some $K \ge 1$. Then $\PP$ is a $\pvec$-equilibrium outcome distribution only if $ \sum_{i=1}^K q_i \PP (q_i, e_i) = \prior$ and
  $e_i \in V(q_i ; \pvec)$ for $i \in \set{1, \dots, K}$. 
  }
A $\pvec$-equilibrium outcome distribution is called an \textit{Expert-optimal $\pvec$-equilibrium outcome distribution} if it is induced by some Expert-optimal $\pvec$-equilibrium.

When $v^*_{\chi} (\prior; \pvec) = \qc v(\prior; \pvec)$, Expert can achieve his $\pvec$-equilibrium value through cheap talk by setting $ξ = σ$.
In this case, there generally exist multiple Expert-optimal $\pvec$-equilibrium outcome distributions.
Suppose $v^*_{\chi} (\prior; \pvec) > \qc v(\prior; \pvec)$.
Then we have $v^*_{\chi} (\prior; \pvec) = \cav  v(\prior; \pvec)$ and 
$\prior \in  [\qlbar (χ ; \pvec) , \qbar(\pvec) )$.
In this case, there is a unique Expert-optimal $\pvec$-equilibrium outcome distribution.
\Cref{prp:experiment} characterizes the $\pvec$-equilibrium outcome distribution and Expert's information disclosure 
in this case.

\begin{proposition}\label{prp:experiment}
If $v^*_{\chi} (\prior; \pvec) > \qc v(\prior; \pvec)$, then the following claims hold.
\begin{enumerate}[(i)]
\item 
  There exists a unique Expert-optimal $\pvec$-equilibrium outcome distribution.
  The support of this outcome distribution
  is $\set{(q_1^*, \s_1^*), (q_2^*, \s_2^*)}$ 
  where $q_1^* = 0$, $\s_1^* = p_1 - c_1$, $q_2^* = \qbar(\pvec)$ and $\s_2^* = p_2 - c_2$.

\item\label{itm:experiment} 
  Expert chooses the experiment $ξ^*$ satisfying
  $$
  ξ^*(m_1 \mid t_1) = \frac{\uchi (\prior ; \pvec) }{χ},
  ξ^*(m_2 \mid t_1) = 1 - \frac{\uchi (\prior ; \pvec) }{χ}
  \text{ and } ξ^*(m_2 \mid t_2) = 1
  $$
  for two distinct messages $m_1$ and $m_2$, and his signalling strategy is
  $\sigma^{\NR} (m_2 \mid t) = 1$ for all $t \in T$.

\end{enumerate}
\end{proposition}

\begin{proof}
  See \Cref{app:prp-experiment}.
\end{proof}

By \Cref{prp:experiment},
when $v^*_{\chi} (\prior; \pvec) > \qc v(\prior; \pvec)$,
a less credible Expert will choose a more informative experiment in that 
$ξ^*(m_2 \mid t_2)$ is independent of $χ$ and that
$ξ^*(m_1 \mid t_1) = \frac{\uchi (\prior ; \pvec) }{χ}$ is decreasing in $χ$.
Specifically, $ξ^*(m_1 \mid t_1) = ξ^*(m_2 \mid t_2) = 1$ when $χ = \uchi (\prior ; \pvec)$, 
meaning that Expert chooses a fully-revealing experiment.
However, a less credible Expert does not disclose more information to Client overall,
as the induced $\pvec$-equilibirum outcome distribution is independent of $χ$.
This is because a less credible Expert has a higher probability of manipulating the experiment result, 
offsetting the effect of a more informative experiment.

\subsection{Expert's equilibrium value}
Call $\pvec$ an \textit{optimal price list} if it solves $\max_{\pvec \in \Pbar \cup P^2} v^*_χ (\prior ; \pvec)$.
\Cref{prp:main} characterizes the optimal price list(s) and Expert's equilibrium value $\ev^*_χ(\prior)$.

\begin{proposition}\label{prp:main}
  Let $χ^*(\prior) \equiv \frac{(\cs - \cm) - \prior (l_2 - l_1)}{(1 - \prior)(\cs - \cm)}$ for $\prior \in (0, 1)$.
  \begin{enumerate}[(i)]
  \item
    If $\prior \le \frac{c_2 - c_1}{l_2 - l_1}$ and $χ \le  χ^*(\prior)$,
    then any price list $\pvec^*$ satisfying $p_1^* = l_1$ and $p_2^* \ge l_1 - c_1 + c_2$ is optimal, and
    $\ev^*_{χ} (\prior) = \qcav v(\prior ; p^*) = l_1 - c_1$. 
  \item\label{itm:unique-op-p}
    Otherwise (i.e., if $χ > \max \{χ^*(\prior) , 0\}$), there exists a unique optimal price list $\pvec^* = (l_1, p_2^* (\prior, χ))$, where
    $p_2^* (\prior, χ)$ is the implicit function determined by 
    $$\uchi (\prior; l_1, p_2^*) = χ,$$
    and Expert's equilibrium value is
    $\ev^*_{χ} (\prior) = \cav v(\prior ; \pvec^*) =  [(l_2 - l_1) - \chi (c_2 - c_1)]\prior + (l_1 - c_1) - (1 - \chi) (c_2 - c_1)$. 
  \end{enumerate}  
\end{proposition}
\begin{proof}
  See \Cref{app:prp-main}.
\end{proof}

\begin{figure}
  \begin{subfigure}{0.49\textwidth}
    \centering
    \includegraphics[width=\textwidth]{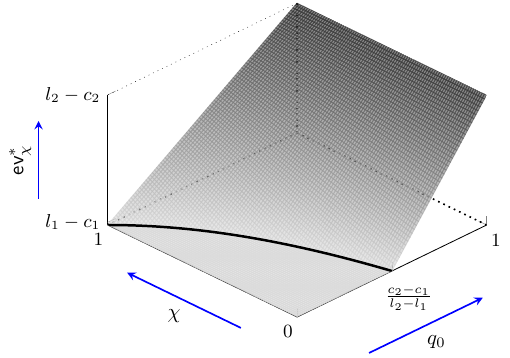}
    \caption{Illustration of $\ev_χ^*(\prior)$} 
    \label{fig:ev}
  \end{subfigure}
  \hfill
  \centering
  \begin{subfigure}{0.45\textwidth}
    \centering
    \includegraphics[]{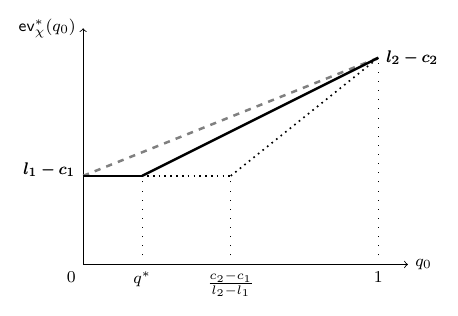}
    \caption{Illustration of $\ev_χ^*(\prior)$ for a fixed $χ$}
    \label{fig:profit-intermediate-chi}
  \end{subfigure}
  \caption{Expert's equilibrium value}
\end{figure}

\Cref{fig:ev} visualizes Expert's equilibrium value $\ev^*_χ(\prior)$
for $χ \in [0,1]$ and $\prior \in [0,1]$.
Unlike $v_χ^*(\prior; \pvec)$, function $\ev^*_χ(\prior)$ is continuous over its entire domain.
Specifically, there is no jump at $χ^*(\prior)$ for $\prior \le \frac{c_2 - c_1}{l_2 - l_1}$, 
denoted by the solid curve on the $\prior$--$χ$ plane.
The reason is that under the pricing scheme $\pvec^*$ as described in
\Cref{prp:main}\ref{itm:unique-op-p},
$\cav v(\prior ; l_1, p_2^* (\prior, χ))$ converges to $l_1 - c_1$ as $χ \searrow χ^*(\prior)$ 
for any $\prior \le \frac{c_2 - c_1}{l_2 - l_1}$.
We refer readers to \Cref{sec:sketch} for explanations on the optimality of the described pricing scheme.

Say that Expert \textit{benefits from credibility} %
if $\ev^*_{χ} (\prior) > \ev^*_{0} (\prior)$.
That is, Expert benefits from credibility if his equilibrium value is strictly higher 
than the benchmark case without credibility.
\Cref{prp:main} implies that
Expert benefits from credibility if and only if $χ > \max \{χ^*(\prior) , 0 \}$.
When $\prior \ge \frac{c_2-c_1}{l_2-l_1}$, $χ^*(\prior) \le 0$ and
Expert benefits from credibility for any positive $χ$.
To make it explicit when (and how much) Expert benefits from credibility,
\Cref{fig:profit-intermediate-chi} illustrates $\ev_χ^*$ for some fixed $χ \in (0,1)$ (the solid curve), for $χ = 1$ (the dashed line) and for $χ = 0$ (the dotted curve),
where the kink $q^*$ is determined by $χ^*(q^*) = χ$.
As shown in \Cref{fig:profit-intermediate-chi}, Expert benefits from credibility 
if and only if $\prior > q^*$, with the benefit $\ev^*_{χ} (\prior) - \ev^*_{0} (\prior)$
maximized at $\prior = \frac{c_2 - c_1}{l_2 - l_1}$.
Notably, the kink $q^*$ approaches 0 (resp.\ $\frac{c_2 - c_1}{l_2 - l_1}$) as $\chi$ approaches $1$  (resp.\ $0$).
Moreover, fixing any $\prior \in (0,1)$, we have
$\lim_{χ \nearrow 1} \ev_χ (\prior) = \ev_1 (\prior)$ and
$\lim_{χ \searrow 0} \ev_χ (\prior) = \ev_0 (\prior)$.
Therefore, regarding Expert's equilibrium value, 
the limiting scenarios of limited credibility as $χ \nearrow 1$ and $χ \searrow 0$
coincide with the extreme cases of $χ=1$ and $χ=0$, respectively.

Expert's information disclosure in equilibrium follows directly from \Cref{prp:experiment,prp:main}.
Specifically, he chooses a fully-revealing experiment whenever $χ > \max\{ χ^*(\prior) , 0 \}$. 

\begin{corollary}\label{cor:signalling}
When the pair $(\prior, χ)$ satisfies $χ > \max\{ χ^*(\prior) , 0 \}$, 
Expert chooses a fully-revealing experiment $ξ^{\FR}$ satisfying  
$ξ^{\FR} (m_1 \mid t_1) = 1$ and $ξ^{\FR} (m_2 \mid t_2) = 1$ 
for two distinct messages $m_1 \text{ and } m_2$,
and his signalling strategy is $\sigma^{\NR} (m_2 \mid t) = 1$ for all $t \in T$.
\end{corollary}

\subsubsection{Optimality of $\pvec^*$\label{sec:sketch}}

We sketch the proof of the optimality of $\pvec^* = (l_1, p_2^* (\prior, χ))$, as specified in \Cref{prp:main}\ref{itm:unique-op-p}.
If $p_2 - c_2 = p_1- c_1$, then it is optimal for Expert to set $p_1 = l_1$ and $p_2 = l_1- c_1 + c_2$ and
his highest ex ante payoff is $l_1 - c_1$.
Now suppose $p_2 - c_2 > p_1- c_1$.
\Cref{fig:p_value} illustrates Expert's $\pvec$-equilibrium value as the prior $\prior$ changes,
where $\qbar^* = \qbar (\pvec)$ and $\qlbar^* = \qlbar (χ; \pvec)$.
While there are two possible cases for computing the kink $\qbar$ as in \Cref{eq:qbar},
the second case, $\qbar = \frac{p_2 - p_1}{l_2}$, is never optimal.
Otherwise, consider Expert increasing both $p_1$ and $p_2$ by some arbitrarily small $ε > 0$.
Then $\qbar (\pvec)$ and $\qlbar (χ ; \pvec)$ remain unchanged and, by \Cref{fig:p_value}, $v^*_{\chi}(\prior ; \pvec)$ strictly increases for all $\prior$.

\begin{figure}
  \centering
  \begin{subfigure}{0.52\textwidth}
      \centering
      \includegraphics[]{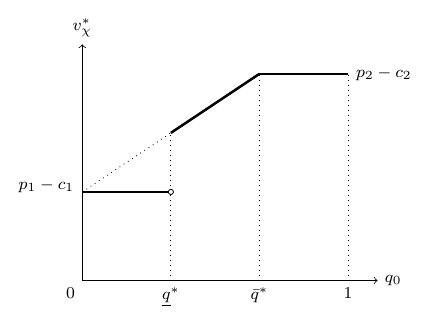}
      \caption{$v_χ^*(\prior ; \pvec)$ for fixed $χ$ and $\pvec$ \label{fig:p_value}}
  \end{subfigure}
  \hfill
  \begin{subfigure}{0.47\textwidth}
      \centering
      \includegraphics{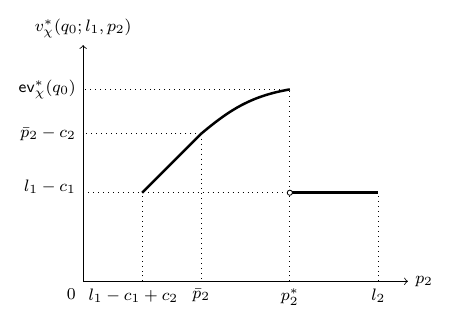}  
      \caption{Optimal $p_2$ for Expert\label{fig:optimal_p2}}      
  \end{subfigure}
  \caption{Illustration of the optimality of $\pvec^*$ when $χ > \max \{χ^* (\prior) , 0 \}$}
\end{figure}

From now on, suppose $p_2 \ge l_2 - \frac{l_2 - l_1}{l_1} p_1$ and then
$\qbar (\pvec)= \frac{\ps - \lm}{\ls - \lm}$.
Since $\qbar(\pvec)$ and $\qlbar(χ ; \pvec)$ are both independent of $p_1$, 
it follows from standard perturbation arguments that
$p_1^* = l_1$ is Expert-optimal.\footnote{See \Cref{app:prp-main} for details.}
Given $p_1 = l_1$, \Cref{fig:optimal_p2} visualizes Expert's $\pvec$-equilibrium value for $p_2 \in [l_1 - c_1 + c_2, l_2]$,
where $\bar{p}_2$ is determined by $\qbar (l_1, \bar{p}_2) = \prior$.
When $p_2 > p_2^*$, we have $χ < \uchi (\prior ; l_1, p_2^*)$ and Expert's payoff is $l_1 - c_1$. 
When $p_2 \in [\bar{p}_2 , p_2^* ]$, Expert's $\pvec$-equilibrium value is $λ(\pvec) \prior + l_1 - c_1$, 
corresponding to the sloping line segment in \Cref{fig:p_value}.
When $p_2 < \bar{p}_2$, Expert's $\pvec$-equilibrium value is $p_2 - c_2$, 
corresponding to the flat line in \Cref{fig:p_value}.
It follows that $v^*_{\chi}( \prior ; l_1, p_2)$ is maximized at $p_2 = p_2^*$.

\section{Welfare analysis}
In this section, we analyze how different credibility levels affect players' welfare.
We also examine the welfare implications of policies regarding the monitoring of expert's disclosure and price regulation.

\subsection{Comparative statics}
The effects of credibility on Expert's equilibrium value follow directly from \Cref{prp:main}.

\begin{corollary} \label{cor:chi}
    Fixing $\prior \in (0, \frac{c_2 - c_1}{l_2 - l_1})$, Expert's equilibrium value
    is independent of $χ$ over $[0 , χ^*(\prior)]$ and is strictly increasing in $χ$ 
    over $[χ^*(\prior) , 1]$;
    fixing $\prior \in [\frac{c_2 - c_1}{l_2 - l_1}, 1)$, Expert's equilibrium value is strictly increasing in $χ$ over $[0,1]$. 
\end{corollary}

We turn to Client's welfare and investigate how it is affected by credibility.
Fixing any pair $(\prior, χ)$, all equilibria yield the same payoff for Expert.
However, Client's welfare can vary among the equilibria.
We use the value of Expert's services for Client to measure Client's welfare:
Define \emph{Client's equilibrium value} as the expected surplus she obtains in equilibrium.
This value can be calculated as the difference between Client's equilibrium payoff and her ex ante payoff from choosing $\an$.
Since the total surplus is $S(\prior) \equiv \prior (l_2 - c_2) + (1- \prior) (l_1 - c_1)$,
Client's equilibrium value cannot exceed $S(\prior) - \ev^*_χ (\prior)$.
When $χ > \max \{χ^*(\prior) , 0\}$, there exists a unique optimal price list $\pvec^*$ and a unique Expert-optimal $\pvec^*$-equilibrium outcome distribution,
in which Client's equilibrium value is always zero.
When $\prior \le \frac{c_2 - c_1}{l_2 - l_1}$ and $χ \le χ^*(\prior)$, multiple optimal price lists exist, and Client's equilibrium values vary across equilibria.
In such instances, Client's equilibrium value can take any value between $0$ and $S(\prior) - \ev_χ^* (\prior)$.

\begin{proposition}\label{prp:client-values}
  If $\prior \in (0, \frac{c_2 - c_1}{l_2 - l_1} ]$ and $χ \in (0, χ^*(\prior)]$, then the set of Client's equilibrium values is 
  $\big[0 , S(\prior) - \ev_χ^* (\prior) \big]$.
\end{proposition}

\begin{proof}
See \Cref{app:prp-client-values}.
\end{proof}

Denote by $\eu^*_χ(\prior)$ \textit{Client's highest equilibrium value} given $\chi$ and $\prior$. 
We have
\[
\eu^*_χ(\prior) = \begin{cases}
  \prior \big((l_2 - c_2) - (l_1 - c_1)\big) & \textif \prior \in  (0, \frac{c_2 - c_1}{l_2 - l_1}] \text{ and } χ \in [0 , χ^*(\prior)] \\
  0 & \text{otherwise}.
\end{cases}
\]
When $\prior \in  (0, \frac{c_2 - c_1}{l_2 - l_1}]$ and $χ \in [0, χ^*(\prior) ]$,
the only equilibrium that achieves Client's highest equilibrium value
is an \textit{equal-margin fully-disclosing equilibrium}:
In this equilibrium, Expert sets $p_1 = l_1$ and $p_2 = l_1 - c_1 + c_2$ and discloses all information to Client.
This kind of equilibrium has been studied extensively in the credence goods literature.%
  \footnote{For discussions on the equal-margin equilibrium, 
    see \cite{dulleck} and the references therein.  
    Many works in the credence goods literature call the expert \textit{honest}
    if the equilibrium is fully-disclosing.}
One appealing feature of the equal-margin fully-disclosing equilibrium is that it maximizes the social welfare.
Still, Client's equilibrium value is not unique for $\prior \in  (0, \frac{c_2 - c_1}{l_2 - l_1}]$ and $χ \in [0, χ^*(\prior) ]$,
even when an equal-margin price list is imposed.
Client's actual equilibrium value and the specific equilibrium to be played will depend on the institutional context.

\subsection{Monitoring expert v.s.\ price regulation} \label{sec:price-policy}
We discuss the policies of monitoring expert (which affects $χ$) and price regulation (which imposes some exogenous price list).
In terms of protecting Client's welfare, our findings do not support the usage of monitoring policies without price regulation.
Increasing monitoring intensity can benefit Expert solely and never improves Client's welfare.
Specifically, when $χ > \max \{χ^*(\prior) , 0 \}$, Client's equilibrium value is always zero.

Our previous analysis has focused on Expert's optimal pricing.
To further examine the effects of price regulation,
it is necessary to characterize Client's welfare under different exogenous price lists.
Denote by $u_χ^*(\prior; \pvec)$ Client's \textit{highest payoff among all Expert-optimal $\pvec$-equilibria}.
We characterize $u_χ^*(\prior; \pvec)$ in \Cref{prp:welfare}. Note that
$\pvec \in \Pbar \cup P^2$ is not imposed, as $\pvec$ need not be Expert-optimal.

\begin{proposition} \label{prp:welfare}
\begin{enumerate}[(i)]
  \item 
    For $\pvec \in P^2$,
    $$
    u_χ^*(\prior; \pvec) = \begin{cases}
      (1 - \frac{\prior}{\qbar(\pvec)}) (l_1 - p_1) + 
      \frac{\prior}{\qbar(\pvec)} [\qbar(\pvec) l_2 + (1-\qbar(\pvec)) l_1 - p_2]
      &\textif \prior \in [0, \qbar(\pvec)) \\
      \prior l_2 + (1-\prior) l_1 - p_2 &\textif \prior \in [\qbar(\pvec), 1]
    \end{cases}
    $$  
  \item 
    For $\pvec \in P^1$,
    $$
    u_χ^*(\prior; \pvec) = \begin{cases}
    (1 - \prior) l_1 - p_1 & \text{ if } \prior \in [0,\tilde{q}(\pvec)] \\
    \frac{1 - \prior}{1 - \tilde{q}(\pvec) } \big ((1 - \tilde{q}(\pvec)) l_1 - p_1 \big) + \frac{\prior - \tilde{q}(\pvec) }{1 - \tilde{q}(\pvec) } (l_2 - p_2) & \text{ if } \prior \in (\tilde{q}(\pvec) , 1]
    \end{cases}
    $$ 
    where 
    $$
    \tilde{q}(\pvec) = \begin{cases}
      \frac{l_1 - p_1}{l_1} & \text{if } p_2 \ge l_2 - \frac{l_2 - l_1}{l_1} p_1 \\
      \frac{p_2 - p_1}{l_2} & \text{otherwise.}
     \end{cases}
    $$
  \item 
    For $\pvec \in \Pbar$,
    $$
    u_χ^*(\prior; \pvec) = \prior (l_2 - p_2) + (1 - \prior) (l_1 - p_1).
    $$
  \end{enumerate}
\end{proposition}

\begin{proof}
See \Cref{app:prp-welfare}.  
\end{proof}

\Cref{prp:welfare} implies that $u_χ^*(\prior; \pvec)$ is independent of $χ$,%
    \footnote{While it is technically accurate to omit the subscript $χ$ in $u_χ^*(\prior; \pvec)$ due to its independence of $χ$, we choose to retain the subscript $χ$ for clarity, as this independence is not obvious beforehand.} %
meaning that monitoring Expert has no effect on
Client's highest equilibrium value when the government can regulate the prices directly.
Moreover, for a fixed $\prior \in (0,1)$, there are
jumps of $u_χ^*(\prior; \pvec)$ over the region $\Pbar$.
Notably, a small decrease in either price may lead to a significant \textit{decrease} in Client's highest equilibrium value
if it discourages Expert from disclosing information.
\Cref{fig:harming-client} illustrates 
how a small decrease in $p_1$ can substantially change
Expert's information disclosure, resulting in harm to Client.
In \Cref{fig:harming-client}(a) where $p_1 - c_1 = p_2 - c_2$, Expert is willing to disclose all information to 
Client, leading to the equilibrium outcomes $(0, p_1 - c_1)$ and $(1, p_1 - c_1)$.
Now consider a small decrease of $ε$ in $p_1$. 
The induced equilibrium outcomes will be either
$\{ (0, p_1 - c_1 - ε), (\qbar, p_1 - c_1 - ε) \}$ as in \Cref{fig:harming-client}(b)
or $\{ (0, p_1 - c_1 - ε), (\qbar, p_1 - c_1) \}$ as in \Cref{fig:harming-client}(c),
depending on whether $χ$ is above $\uchi (\prior ; p_1 - ε, p_2)$.
In both situations, Client's payoff will drop significantly.

\begin{figure}[!htbp]
  \makebox[\linewidth][c]{%
    \begin{subfigure}{0.4\textwidth}
      \centering
      \includegraphics[width=\linewidth]{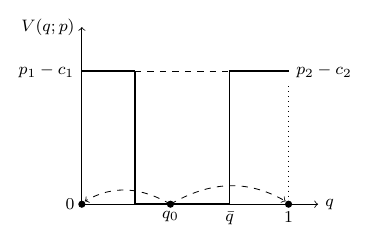}
      \caption{\footnotesize Full disclosure} %
    \end{subfigure}
    $\,$
    \begin{subfigure}{0.4\textwidth}
      \centering
      \includegraphics[width=\linewidth]{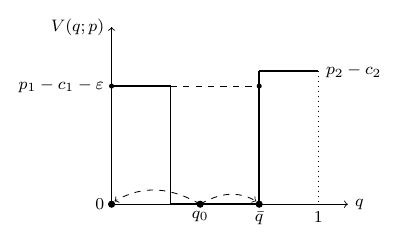}
      \caption{\footnotesize $χ < \uchi (\prior ; p_1 - ε, p_2)$}
    \end{subfigure}
    $\,$
    \begin{subfigure}{0.4\textwidth}
      \includegraphics[width=\linewidth]{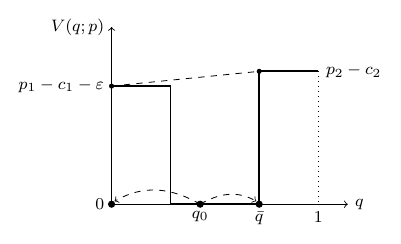}
      \caption{\footnotesize $χ \ge \uchi (\prior ; p_1 - ε, p_2)$}
    \end{subfigure}%
  }  
    \caption{How decreasing $p_1$ can harm Client\label{fig:harming-client}}
\end{figure}

This observation emphasizes the superiority of the equal-margin fully-disclosing equilibrium in
enhancing social welfare. \Cref{cor:discontinuity} formalizes this observation,
where $u^*_χ | _{P^i}$ is the restriction of $u^*_χ$ to  $P^i$ for $i \in \set{1,2}$.

\begin{corollary}\label{cor:discontinuity}
  If $\pbar \in \Pbar$, then
  $u_χ^*(\prior; \pbar) > \limsup\limits_{p \to \pbar} u^*_χ | _{P^1} (\prior; \pvec)$ 
  and $u_χ^*(\prior; \pbar) > \limsup\limits_{p \to \pbar} u^*_χ | _{P^2}(\prior; \pvec)$
  for all $\prior \in (0,1)$ and $χ \in (0,1)$.
\end{corollary}

\begin{proof}
  See \Cref{app:cor-discontinuity}.
\end{proof}

\section{Discussions} \label{sec:discuss}

We discuss two variants of our model.
In the first variant,  
Expert cannot observe the problem type directly.
In the second variant,
Client also observes Expert's credibility type when it is realized.

\subsection{Unknown problem type\label{sec:discussion-type}}

In the main analysis, we have assumed that
Expert privately observes Client's problem type
upon her visit.
An alternative setting worth exploring
is that Expert cannot observe Client's problem type directly. 
In this case, Expert's private information will be the message generated by his chosen experiment 
instead of the problem type. 
Below we argue that this modification does not affect our main results.

The characterization of Expert's equilibrium value will be the same as in \Cref{prp:main}.
This can be obtained by  comparing our main model 
$\cG$ with the alternative game, denoted by $\cG^A$, where the only difference is Expert's signaling strategy being restricted to $σ^A: M \to \Delta (M)$. 
It can be verified that 
a $\pvec$-equilibrium of $\cG^A$ is essentially a
$\pvec$-equilibrium of $\cG$ (see \Cref{app:unknown-type}). 
As a result, Expert's equilibrium value of
$\cG^A$
is weakly lower than that of $\cG$.
On the other hand, in our main model $\cG$
Expert 
can always use a fully-revealing experiment
to achieve his equilibrium value $\ev^*_χ (\prior)$.%
    \footnote{Specifically, when $χ > \max \{χ^*(\prior) , 0 \}$, Expert always chooses a fully-revealing experiment (\Cref{cor:signalling});
    otherwise, $\ev^*_χ (\prior) = l_1 - c_1$ and
    there exists an equal-margin fully-disclosing equilibrium, in which Expert chooses a fully-revealing experiment.}
Since Expert can learn the problem type from the experiment result, these equilibria of $\cG$ can also be seen as equilibria of $\cG^A$.
Consequently, Expert's equilibrium value of $\cG^A$ remains $\ev^*_χ (\prior)$.

Furthermore, the characterization of Client's equilibrium values will be the same as that of our main model.
When $χ > \max \{χ^*(\prior) ,0 \}$, both the optimal price list $\pvec^*$ and the Expert-optimal $\pvec^*$-equilibrium outcome distribution 
are unique in our main model.
Since the alternative game $\cG^A$ has fewer $\pvec$-equilibria,
the optimal price list $\pvec^*$ and the Expert-optimal
outcome distribution are also unique in $\cG^A$.
As a result, the set of Client's equilibrium values in $\cG^A$ is $\set{0}$.
When $\prior \le \frac{c_2 - c_1}{l_2 - l_1}$ and $χ \le χ^*(\prior)$,
we have focused on those equilibria in which $\sigma = \xi$ to characterize the set of Client's equilibrium values in our main model.
The same argument can be applied to $\cG^A$, and it follows that
the set of Client's equilibrium values in $\cG^A$ is $[0, S(\prior) - \ev_χ^* (\prior)]$.

\subsection{Public credibility \label{sec:discuss-public-credibility}}

Suppose Expert's credibility type is publicly revealed when it is realized. 
Then, when Expert is revealed credible, it is common knowledge that the message is from a credible source and Expert's highest payoffs are $\cav v(\prior; \pvec)$.
When Expert is revealed non-credible,
it is common knowledge that Expert can manipulate the message and his     
highest payoffs are $\qcav v(\prior; \pvec)$.
So Expert's $\pvec$-equilibrium value is:
\begin{equation*} %
  v_{χ}^{\pc}(\prior; \pvec) = (1 - χ)\qcav v(\prior; \pvec) +  χ \cav v(\prior; \pvec).
\end{equation*}
Expert chooses some price list to maximize $v_{χ}^{\pc}(\prior; \pvec)$.
In \Cref{app:public-credibility},
we show that an Expert-optimal price list is 
$p_1^{\pc} = l_1$ and
$$
p_2^{\pc} (\prior, χ) = \begin{cases}
  l_2 
    & \textif  \prior < q_\chi^{\pc}\\
  l_2 \text{ or } \prior l_2 + (1-\prior) l_1 
    & \textif  \prior = q_\chi^{\pc}\\
  \prior l_2 + (1-\prior) l_1
    & \textif  \prior > q_\chi^{\pc}
\end{cases}  
$$
where 
$q_\chi^{\pc} = \frac{c_2 - c_1}{(1 - \chi) (l_2 - l_1) + \chi (c_2 - c_1)}$.
Correspondingly, Expert's equilibrium value is
$$
\ev_{χ}^{\pc} (\prior) %
= \begin{cases}
  \prior l_2 + (1-\prior) l_1 - c_2
    & \textif  \prior \ge q_\chi^{\pc} \\
  \chi \prior (l_2 - c_2) + (1 - \chi \prior) (l_1 - c_1) 
    & \textif  \prior < q_\chi^{\pc}
\end{cases}
$$

\begin{figure}[!htbp]
  \centering
  \includegraphics{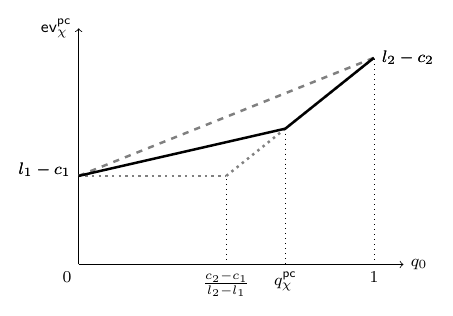}  
  \caption{Expert's equilibrium value given public credibility\label{fig:modification}}
\end{figure}

\Cref{fig:modification} illustrates Expert's equilibrium value for some fixed $χ \in (0, 1)$.
As comparisons, 
\Cref{fig:modification} also 
illustrates the case for $χ = 0$ (the dotted curve) 
and the case for $χ = 1$ (the dashed line). 
Note that the kink $q_χ^{\pc}$ converges to $1$ (resp.\ $\frac{c_2 - c_1}{ l_2 - l_1 }$) as $χ$ converges to $1$
(resp.\ $0$).

Allowing public credibility affects the characterizations of Expert's and Client's equilibrium values.
Notably, we find that the comparison between
$\ev_χ^* (\prior)$ and $\ev_χ^{\pc} (\prior)$ depends solely on the prior.
When $\prior > \frac{c_2 - c_1}{l_2 - l_1}$ (resp.\ $\prior < \frac{c_2 - c_1}{l_2 - l_1}$), we have $\ev_χ^* (\prior) > \ev_χ^{\pc} (\prior)$ (resp.\ $\ev_χ^* (\prior) < \ev_χ^{\pc} (\prior)$)
for any $χ \in (0,1)$.
In particular, when $\prior = \frac{c_2 - c_1}{l_2 - l_1}$, we have
$\ev_χ^* (\prior) = \ev_χ^{\pc} (\prior)$  for any $χ \in (0,1)$. 
Additionally, in the case of public credibility,
Client's equilibrium value is always zero, and an equal-margin fully-disclosing equilibrium does not exist for any $(χ,\prior) \in (0,1) \times (0,1)$.

\bibliographystyle{jpe}
\bibliography{credibility.bib}

\clearpage

\appendix

\section{Appendix}

The Appendix contains the omitted proofs.

\subsection{Proof of \Cref{lmm:1} \label{app:lmm-1}}

Fix some $\pvec = (p_1, p_2) \in P^1$, and it follows that $v_χ^* (\prior ; \pvec) \le \pm - \cm$.
We further show that
\begin{equation}
  v_χ^* (\prior ; \pvec) < l_1 - c_1 \text{ for all } \prior \in (0,1). \label{eq:lmm1}  
\end{equation}
When $v_χ^* (\prior ; \pvec) < \pm - \cm $,
condition \eqref{eq:lmm1} holds as $p_1 \le l_1$;
When $v_χ^* (\prior ; \pvec) = \pm - \cm$,
Client must choose $\am$ in all possible outcomes and
$v_χ^* (\prior ; \pvec) \le \prior l_1 - c_1 < l_1 - c_1$.

Consider another price list $\pvec' = (l_1, l_1 - \cm + \cs)$.
There always exists a $\pvec'$-equilibrium in which Expert discloses all information to Client (i.e., both $\sigma$ and $\xi$ are fully-revealing) and in which
Client chooses $a_i$ given the problem type $t_i$ for $i \in \{1,2\}$. 
Therefore,
$$
v_χ^* (\prior ; \pvec') \ge l_1 - c_1 \text{ for all } \prior \in (0,1).
$$

\subsection{Proof of \Cref{prp:p-eq-value} \label{app:prp-p-eq-value}}

Fixing some $(\prior, \pvec)$ satisfying $\pvec \in P^2$ and $\prior < \qbar(\pvec)$,
we calculate Expert's $\pvec$-equilibrium value $v^*_{\chi}(\prior ; \pvec)$.
In a seminal paper, \cite{lipnowski2022} provide a recipe for calculating $v^*_{\chi}(\prior ; \pvec)$ in a more general setting than our main model.
For the sake of completeness, we restate some of their results here.
Denote by $\PP \in \Delta ([0,1] \times \RR)$ the ex ante joint distribution of Client's posterior and Expert's expected payoff.
Refer to $(q, \s) \in [0,1] \times \RR$ as an \textit{outcome of $\cG_{\pvec}$}
and refer to $\PP$ as an \textit{outcome distribution of $\cG_{\pvec}$}.
Let $\smax \equiv \max \{\s : (q, \s) \in \supp (\PP) \}$
be Expert's maximal payoff under distribution $\PP$.
We decompose $\PP$ into two distributions---the distribution over outcomes conditional on $\s =\smax$ (denoted by $\GG$)
and the distribution over outcomes conditional on $\s < \smax$ (denoted by $\BB$).
It follows that $\PP = (1 - k) \GG + k \BB$,
where $k \in [0,1]$ is the probability that Expert's payoffs are strictly below $\smax$.
Let $γ$ and $β$ be the mean posterior of problem type $t_2$ conditional on $\GG$ and $\BB$ respectively.
Then,
\begin{equation}
  (1 - k) γ + k β = \prior.
  \tag{BS} 	
  \label{eq:BS}
\end{equation}

Besides the Bayesian-splitting constraint~\eqref{eq:BS}, there exists an additional constraint stemmed from the limited credibility of Expert.
Conditional on Expert being non-credible, he will always send some favorable message which leads to the highest payoff $\smax$.
So, for each $t \in T$, the probability of the event
$$ E_1 =
  \{ 
    \text{The problem type is $t$ and Expert's payoff is $\smax$ in the outcome}
  \}
$$
is at least the probability of the event
$$
E_2 = \{
  \text{The problem type is $t$ and Expert is non-credible}
\}.
$$
It follows that
\begin{equation} \label{eq:chi} \tag{$\chi^C$}
  (1 - k) γ \geq \prior (1 - \chi), \quad 
  (1 - k) (1- γ)  \geq  (1- \prior) (1 - \chi).
\end{equation}

Define an auxiliary correspondence that caps the correspondence
$V(q ; \pvec)$ by some $\bar \s \in \RR$:
$V_{\wedge \bar \s} (q;\pvec) = \{\min \{ \s, \bar \s \}: \s \in V(q ; \pvec) \}$.
Let $v_{\wedge γ} (q; \pvec) \equiv \max V_{\wedge \qc v (γ ; \pvec)}(q ; \pvec)$.
\citeauthor{lipnowski2022} (2022, Theorem~1)
show that Expert's $\pvec$-equilibrium value can be calculated as the value function of the following maximization problem:
\begin{align}
    \label{prg:v} \tag{$\mathcal{M}$}
	  v^*_{\chi}(\prior; \pvec) =
	  & \max_{β, γ , k \in [0,1]} 
	  {(1 - k) \, \qc v (γ ; \pvec) + k \, \cav v_{\wedge  γ} (β; \pvec) }\\
	  & \text{ subject to } \eqref{eq:BS} \text{ and } \eqref{eq:chi}. \nonumber
\end{align}
That is, Expert chooses some $(β^*, γ^*, k^*)$ to maximize his objective function 
$(1 - k) \qc v (γ ; \pvec) + k \cav v_{\wedge  γ} (β; \pvec)$,
subject to constraints \eqref{eq:BS} and \eqref{eq:chi}.

We first calculate $\cav v_{\wedge γ}( q ; \pvec)$ as $γ$ varies.
Note that $\qcav v(γ ; \pvec)$ takes the value of either $p_1 - c_1$ or $p_2 - c_2$.
When $\qcav v(γ ; \pvec) = p_2 - c_2$,
$\cav v_{\wedge γ}(q ; \pvec)$ is the same as $\cav v( q ; \pvec)$;
When $\qcav v(γ ; \pvec) = p_1 - c_1$,
$\cav v_{\wedge γ}(q ; \pvec)$ takes the constant value of $p_1 - c_1$. To sum up, 
$$
\cav v_{\wedge γ} (q ; \pvec) = \begin{cases}
  \cav v(q ; \pvec) & \text{if } γ \ge \qbar(\pvec); \\
  p_1 - c_1     & \text{if } γ < \qbar(\pvec).
\end{cases}
$$
Below we show that Expert always chooses $γ^* = \qbar(\pvec)$ in the solution to \eqref{prg:v}
whenever his $\pvec$-equilibrium value is strictly higher than $\qc v (\prior ; \pvec) = p_1 - c_1$.

\begin{claim}\label{clm:2}
  If $v^*_{\chi}(\prior; \pvec) > \qc v (\prior ; \pvec)$, then 
  $γ^* = \qbar(\pvec)$ in the solution to \eqref{prg:v}.
\end{claim}

\begin{proof}
Since $v^*_{\chi}(\prior; \pvec) > \qc v (\prior ; \pvec)$,
$γ \ge \qbar(\pvec)$ in any $\pvec$-equilibrium.
Further, we have $\cav v(β ; \pvec) = \lambda (\pvec) β + p_1 - c_1$ where $\lambda (\pvec) = \frac{p_2 - c_2 - (p_1 - c_1)}{\qbar(\pvec)}$.
Expert's maximization problem \eqref{prg:v} can be rewritten as below. 
\begin{equation} \label{prg:v'} \tag{$\mathcal{M}'$}
  v^*_{\chi}(\prior; \pvec) =
  \max_{β, γ, k \in [0,1]} 
  {(1 - k) \, (p_2 - c_2) + k \big( \lambda (\pvec) β + p_1 - c_1  \big)}
\end{equation}\vspace{-1.0cm}
\begin{align*}
  \text{subject to } 
  & \label{eq:bs-b}  (1 - k) γ + k β = \prior, \tag{BS}
  \\
  \label{eq:chi1_2} \tag{$\chi_1^C$} & k β \leq \chi \prior, \\
  \label{eq:chi2_2} \tag{$\chi_2^C$} & k (1 - β) \leq \chi (1 - \prior).
\end{align*}

Let $(β^*, γ^*, k^*)$ be a solution to \eqref{prg:v'}. Assume by contradiction that $γ^* > \qbar(\pvec)$.
Consider another $γ = γ^* - ε$ for some $ε > 0$. Meanwhile, $β$ is fixed at $β^*$, and
$k$ is adjusted accordingly such that constraint~\eqref{eq:bs-b} is satisfied:
\begin{equation} \label{eq:k}
k (β^*, γ) = 
    \frac{γ - \prior}{γ - β^*}.
\end{equation}
Since $γ < γ^*$, we have $k (β^*, γ) < k^*$ by \Cref{eq:k}.
A lower $k$ relaxes the constraints \eqref{eq:chi1_2} and \eqref{eq:chi2_2} and leads
to a strictly higher objective value in \eqref{prg:v'}.
The objective value gets higher because (i)
$p_2 - c_2 >  \lambda (\pvec) β^* + p_1 - c_1$ and (ii)
$\lambda (\pvec)$ is independent of $k$.
Therefore, lowering $γ^*$ by some arbitrarily small $ε > 0$ strictly increases the objective value.
Contradiction.
\end{proof}

If $v^*_{\chi}(\prior; \pvec) > \qc v (\prior ; \pvec)$,
we can simplify the problem \eqref{prg:v'} by plugging $γ = \qbar(\pvec)$ 
into the objective function and the constraints.
Below we show that this will lead to $\chi \geq \uchi ( \prior ; \pvec)$ and
$v^*_{\chi}(\prior ; \pvec) = \cav v (\prior ; \pvec)$.

\begin{claim}
  If $v^*_{\chi}(\prior; \pvec) > \qc v (\prior ; \pvec)$, then
  $v^*_{\chi} (\prior ; \pvec) = \cav v (\prior ; \pvec)$ and 
  $\chi \geq \uchi ( \prior ; \pvec)$.
\end{claim}

\begin{proof}
  Suppose $v^*_{\chi}(\prior; \pvec) > \qc v (\prior ; \pvec)$.
  By \Cref{clm:2}, it's optimal for Expert to set
  $γ = \qbar(\pvec)$,
  and constraint \eqref{eq:bs-b} implies $β = \qbar(\pvec) - \qbar(\pvec)/k + \prior /k$.
  Plugging $γ = \qbar(\pvec)$ and $β = \qbar(\pvec) - \qbar(\pvec)/k + \prior /k$ 
  into the objective function of problem \eqref{prg:v'} yields:
\begin{align*}
  v^*_{\chi}(\prior ; \pvec) = λ(\pvec) \prior + (p_1 - c_1).
\end{align*}
Thus, the objective function is exactly $\cav v (\prior ; \pvec)$ and is independent of $k$.
On the other hand, the constraints of \eqref{prg:v'}
reduce to
\begin{align}
  & k \le 1 - (1 - \chi) \frac{\prior}{\qbar(\pvec)}, \label{eq:c1} \\
  & k \le 1 - (1 - \chi) \frac{1 - \prior}{1 - \qbar(\pvec)} \label{eq:c2}\\  
  &
  β = \qbar(\pvec) - \frac{\qbar(\pvec)}{k} + \frac{\prior}{k} \in [0,1].  \label{eq:c3}  
\end{align}
Our prerequisite  $v^*_{\chi}(\prior; \pvec) > \qc v (\prior ; \pvec)$ holds
only if there exists some $k \in [0,1]$ such that constraints (\ref{eq:c1}--\ref{eq:c3}) hold simultaneously.
Constraint~\eqref{eq:c1} is redundant given constraint~\eqref{eq:c2}, 
and constraints~\eqref{eq:c2} and \eqref{eq:c3} imply
that $\chi \geq \uchi ( \prior ; \pvec) \equiv \frac{\qbar(\pvec) - \prior}{ \qbar(\pvec) (1 - \prior)}$.
\end{proof}

We conclude the proof by showing that the condition $\chi \ge \uchi ( \prior ; \pvec)$
also implies $v^*_{\chi}(\prior; \pvec) > \qc v (\prior ; \pvec)$.

\begin{claim}
  If $\chi \ge \uchi ( \prior ; \pvec)$, then 
  $v^*_{\chi}(\prior; \pvec) > \qc v (\prior ; \pvec)$.
\end{claim}

\begin{proof}
Suppose $\prior < \qbar(\pvec)$ and $\chi \ge \uchi ( \prior ; \pvec)$.
Then $\qcav v (\prior ; \pvec) = p_1 - c_1$. 
Consider $(β^*, γ^*, k^*) = (0,  \qbar(\pvec) , 1 - \prior / \qbar(\pvec))$,
which is feasible to the maximization problem \eqref{prg:v'}.
The objective value of \eqref{prg:v'} under $(β^*, γ^*, k^*)$ is $\lambda (\pvec) \prior + p_1 - c_1$.
Therefore, $v^*_{\chi}(\prior ; \pvec) \ge \lambda (\pvec) \prior + p_1 - c_1 > p_1 - c_1 = \qcav v (\prior ; \pvec)$.
\end{proof}

\subsection{Proof of \Cref{prp:experiment}\label{app:prp-experiment}}

Suppose $v^*_{\chi}(\prior; \pvec) > \qc v (\prior ; \pvec)$.
By \Cref{prp:p-eq-value}, we have
$v^*_{\chi}(\prior; \pvec) = \cav v (\prior ; \pvec)$ and $\prior < \qbar(\pvec)$.
To achieve the value of $\cav v (\prior ; \pvec)$ for Expert, the  
outcome distribution $\PP^*$ must satisfy
$$
\supp \PP^* = \set{ (q_1^* , v(q_1^* ; \pvec)), (q_2^* , v(q_2^* ; \pvec))}
\text{where $q_1^* = 0$ and $q_2^* = \qbar(\pvec)$.}
$$
This condition and the Bayesian-splitting constraint
determine the outcome distribution, resulting in a unique Expert-optimal $\pvec$-equilibrium distribution.

Denote by $m_1$ and $m_2$ the two distinct messages that induce posteriors $q_1^*$ and $q_2^*$ respectively.
Since $q_2^*$ leads to a strictly higher payoff for Expert,
$σ (m_2 | t_1) = σ (m_2 | t_2) = 1$.
Since $q_1^* = 0$, Expert sends message $m_1$ only when
$t = t_1$ and $ξ^* (m_2 \mid t_2) = 1$.
The value of $ξ^* (m_1 \mid t_1)$ can be obtained from Client's belief updating.
Specifically, let $ξ_{11} \equiv ξ^* (m_1 \mid t_1)$. Then
$$
\frac{ \prior }{\prior + (1 - χ) (1 - \prior) + (1 - ξ_{11}) (1 - \prior) χ }  =  \qbar (\pvec)
$$
yields
$$
ξ_{11} = \frac{\qbar(\pvec) - \prior}{χ \qbar(\pvec) (1 - \prior)} = 
\frac{\uchi ( \prior ; \pvec)}{χ}.
$$

\subsection{Proof of \Cref{prp:main}\label{app:prp-main}}

By \Cref{lmm:1},
we focus on $\pvec \in \Pbar \cup P^2$. 
When $\pvec \in \Pbar$, 
Expert's optimal prices are $p_1 = l_1$ and $p_2 = l_1- c_1 + c_2$.
Then his ex ante payoff is $l_1 - c_1$ for all $χ$ and $\prior$.

Suppose $\pvec \in P^2$.
There exist two possible functional forms of $\qbar(\pvec)$ as in \Cref{eq:qbar}, depending on whether
$p_2$ is above $l_2 - \frac{l_2 - l_1}{l_1} p_1$.
Below we show that setting $p_2 < l_2 - \frac{l_2 - l_1}{l_1} p_1$ is never optimal.

\begin{claim}
It is never optimal for Expert to choose $\pvec$ satisfying
$p_2 < l_2 - \frac{l_2 - l_1}{l_1} p_1$.
\end{claim}

\begin{proof}
Assume $p_2 < l_2 - \frac{l_2 - l_1}{l_1} p_1$, and then $\qbar (\pvec) = \frac{p_2 - p_1}{l_2}$.
Condition $p_2 < l_2 - \frac{l_2 - l_1}{l_1} p_1$  also implies 
$p_1 < l_1$ and $p_2 < l_2$.
Consider that Expert increases both $p_1$ and $p_2$ by some arbitrarily small $ε > 0$.
By doing that, $\qbar  (\pvec)$ and $\qlbar (χ ; \pvec)$ are unaltered and
$v^*_{\chi}(μ_0 ; \pvec)$ strictly increases for all $\prior \in (0,1)$.
\end{proof}

From now on, suppose $p_2 \ge l_2 - \frac{l_2 - l_1}{l_1} p_1$. 
Then $\qbar (\pvec)= \frac{\ps - \lm}{\ls - \lm}$, which is independent of $p_1$.
Furthermore, $\qlbar (χ, \pvec) = \frac{\qbar (\pvec) (1 - χ)}{ 1 - χ \qbar (\pvec)}$.
Expert's $\pvec$-equilibrium value can be written as:
\begin{equation}\label{eq:optm-p-value}
  v^*_{\chi}(\prior; \pvec) = \begin{cases}
    p_2 - c_2 & 
      \textif \prior > \qbar (\pvec) \\
    \lambda(\pvec) \prior + p_1 - c_1  & 
      \textif \prior \in [\qlbar (χ ; \pvec) , \qbar(\pvec) ]   \\
    p_1 - c_1  & 
      \textif  \prior < \qlbar (χ ; \pvec)
  \end{cases}  
\end{equation}
where $\lambda(\pvec) = \frac{p_2 - c_2 - (p_1 - c_1)}{\qbar(\pvec)}$.
Below we show that it is never optimal for Expert to set $p_1 < l_1$.

\begin{claim}\label{clm:xiong}
  For all $\pvec \in \Pbar \cup P^2$ satisfying $p_1 < l_1$ and $p_2 \ge l_2 - \frac{l_2 - l_1}{l_1} p_1$, 
  there exists  $\pvec' \in P$ such that $v_χ^* (\prior ; \pvec') > v_χ^* (\prior ; \pvec)$.
\end{claim}

\begin{proof}
  Fix some $ \pvec = ( p_1,  p_2) \in \Pbar \cup P^2$ satisfying $ p_1 < l_1$ and $ p_2 \ge l_2 - \frac{l_2 - l_1}{l_1}  p_1$.
  We want to find some $\pvec'$ such that $v_χ^* (\prior ; \pvec') > v_χ^* (\prior ;  \pvec)$.
  
  When $\prior > \qbar( \pvec)$, consider $\pvec' = ( p_1,  p_2 + ε)$ 
  where the positive $ε$ is sufficiently small such that $\prior > \qbar (\pvec')$.
  By \Cref{eq:optm-p-value},
  $v_χ^* (\prior ; \pvec') =  p_2 + ε - c_2 >  p_2 - c_2 = v_χ^* (\prior ; \pvec)$.
  When $\prior < \qbar(\pvec)$, consider $\pvec' = (p_1 + ε , p_2 )$ for some small $ε > 0$.
  Then $\qbar (\pvec') = \qbar(\pvec)$ and $\qlbar (\pvec', χ) = \qlbar(\pvec, χ)$.
  It follows that
  $v_χ^* (\prior ; \pvec') > v_χ^* (\prior ; \pvec)$.

  Finally, suppose $\prior = \qbar(\pvec)$ and then $v_χ^* (\prior ; \pvec) = p_2 - c_2$. 
  Let $\pvec' = (l_1 , p_2 + ε)$ for some $ε > 0$.
  When $l_1 - c_1 \ge p_2 + ε - c_2$, Expert's $\pvec'$-equilibrium value is at least
  $p_2 + ε - c_2$. 
  When $l_1 - c_1 < p_2 + ε - c_2$,
  we have $\qbar (\pvec') > \qbar (\pvec) = \prior$, and 
  $\qlbar (\pvec' , χ) < \prior$ as long as $ε$ is sufficiently small.
  Then 
  \begin{equation*}
    \begin{split}
      v_χ^* (\prior ; \pvec') - v_χ^* (\prior ; \pvec)
      & = λ(\pvec') \qbar(\pvec) + l_1 - c_1 - (p_2 - c_2) \\
      & = (p_2 - c_1) - (c_2 - c_1) \big[\frac{p_2 - l_1}{p_2 + ε - l_1}\big] - (p_2 - c_2)  \\
      & = (c_2 - c_1) \frac{ε}{p_2 + ε - l_1} > 0.
    \end{split}  
  \end{equation*}
\end{proof}

Let $χ^* (\prior) = \frac{c_2 - c_1 - \prior (l_2 - l_1)}{ (c_2 - c_1) (1 - \prior)}$.
Below we show that Expert cannot obtain a payoff strictly higher than $l_1 - c_1$ when
$χ \le χ^* (\prior)$ and $\prior \le \frac{c_2 - c_1}{l_2 - l_1}$.

\begin{claim}
  If $\prior \le \frac{c_2 - c_1}{l_2 - l_1}$ and $χ \le χ^* (\prior)$,
  then $\ev_χ(\prior) = l_1 - c_1$.
\end{claim}

\begin{proof}
  We focus on price lists satisfying 
  $p_1 = l_1$ (by \Cref{clm:xiong}) and $\pvec \in \Pbar \cup P^2$.  
  It follows that $p_2 \in [l_1 - c_1 + c_2 , l_2]$ and $\qc v (\prior ; \pvec) = l_1 - c_1$. 
  Let 
  \[ f(p_2) \equiv
  \uchi(\prior ; l_1, p_2) = \frac{\qbar(l_1, p_2) - \prior}{\qbar (l_1, p_2)(1-\prior)} =
  \frac{ (p_2 - l_1) - \prior  (l_2 - l_1) }{ (p_2 - l_1) (1 - \prior)}.
  \]
  Since $f(\wcard)$ is strictly increasing over $[l_1 - c_1 + c_2 , l_2]$,
  we have $\uchi (\prior ; l_1, p_2) > f(l_1 - c_1 + c_2) = \frac{c_2 - c_1 - \prior (l_2 - l_1)}{ (c_2 - c_1) (1 - \prior)} = χ^*(\prior)$ for all
  $p_2 > l_1 - c_1 + c_2$.
  In this case,
  $χ \le χ^* (\prior) < \uchi (\prior ; l_1, p_2)$ and thus
  $v_χ^* (\prior ; l_1, p_2) = \qc v(\prior ; l_1, p_2) = l_1 - c_1$. 
  When $p_2 = l_1 - c_1 + c_2$,
  Expert's $\pvec$-equilibrium value remains $l_1 - c_1$. 
  To sum up, $v^*_χ (\prior; l_1, p_2) = l_1 - c_1$ for all $p_2 \in [l_1 - c_1 + c_2, l_2]$.
\end{proof}

We solve for the optimal price list when $χ > \max \{χ^*(\prior) , 0\}$.

\begin{claim}
  Given $p_1 = l_1$, there exists a unique optimal $p_2^*$ determined by $\uchi (\prior ; l_1, p_2^*) = χ$ 
  when $χ > \max \{χ^*(\prior) , 0\}$.
\end{claim}

\begin{proof}
Suppose $χ > \max \{χ^*(\prior) , 0\}$. Then there exists $p_2 \in (l_1 - c_1 + c_2, l_2]$ such that 
$χ \ge \uchi (\prior ; l_1, p_2)$ and
$v_χ^* (\prior ; l_1, p_2) = \cav v (\prior ; l_1, p_2) > l_1 - c_1$.
Since $\uchi (\prior ; l_1, p_2)$ is strictly increasing in $p_2$,
there exists a unique $p_2^* = l_1 + \frac{\prior (l_2 - l_1)}{1 - χ + χ \prior}$ such that $\uchi (\prior ; l_1, p_2^*) = χ$.
We verify that $v^*_{\chi}(\prior ; l_1 , p_2)$ achieves its maximum at $p_2 = p_2^*$.

Suppose $p_2 \in [l_1 - c_1 + c_2 , p_2^*]$,
we have $v^*_{\chi}(\prior; l_1, p_2) = \cav v (\prior ; l_1, p_2)$. 
The expression of $\cav v (\prior ; l_1, p_2)$ can be written out as below.
When $\prior > \frac{c_2 - c_1}{l_2 - l_1}$, we have
\begin{equation*}
  \cav v (\prior ; l_1 , p_2) = 
  \begin{cases}
    p_2 - c_2 & \textif p_2 \in [l_1 - c_1 + c_2, \bar {p}_2] \\
    l_1 - c_1 + (l_2 - l_1) \prior - \frac{(c_2 - c_1) (l_2 - l_1)}{p_2 - l_1} \prior & \textif p_2 \in (\bar {p}_2, p_2^*]
  \end{cases}
\end{equation*}
where $\bar{p}_2 = \prior l_2 + (1 - \prior) l_1$ is determined by $\qbar (l_1, \bar {p}_2) = \prior$;
When $\prior \leq \frac{c_2 - c_1}{l_2 - l_1}$, we have
$\cav v (\prior ; l_1 , p_2) = 
l_1 - c_1 + (l_2 - l_1) \prior - \frac{(c_2 - c_1) (l_2 - l_1)}{p_2 - l_1} \prior$.
In both cases, $\cav v (\prior ; l_1 , p_2)$ is strictly increasing in $p_2$ over $[l_1 - c_1 + c_2, p_2^*]$.
So $v^*_{\chi}(\prior; l_1, p_2) < v^*_{\chi}(\prior; l_1, p_2^*)$ for all $p_2 \in [l_1 - c_1 + c_2 , p_2^*)$.

Suppose $p_2 > p_2^*$, we have $v^*_{\chi}(\prior; \pvec) = \qcav v (\prior ; \pvec)$.
Since $\qbar (l_1, p_2) > \qbar (l_1, p_2^*) = \frac{\prior}{1 - χ + χ\prior} > \prior$,
we have $\qcav v (\prior ; \pvec) = l_1 - c_1$.
Therefore, $v^*_{\chi}(\prior; l_1, p_2) < v^*_{\chi}(\prior; l_1, p_2^*)$ for all $p_2 \in (p_2^* , l_2]$.
\end{proof}

Correspondingly, Expert's equilibrium value given $χ > \max \{χ^*(\prior) , 0\}$ is
$\ev (\prior) = \cav v(\prior ; l_1, p_2^*) = [(l_2 - l_1) - \chi (c_2 - c_1)]\prior + (l_1 - c_1) - (1 - \chi) (c_2 - c_1)$.

\subsection{Proof of \Cref{prp:client-values}} \label{app:prp-client-values}

Fixing some pair $(\prior , χ)$ satisfying $0 < \prior \le \frac{c_2 - c_1}{l_2 - l_1}$ and
$0 < χ \le χ^*(\prior)$, then
$\ev_χ^* (\prior) = \qcav v (\prior ; p^*) = l_1 - c_1$.
Consider an optimal price list $p^* = (l_1, l_1 - c_1 + c_2)$ and
focus on those equilibria in which $σ = \xi$.
Let $\set{q_1, q_2}$ denote the collection of Client's possible priors in an $p^*$-equilibrium.
It follows that for any $q_2 \in [\frac{c_2 - c_1}{l_2 - l_1}, 1]$, there exists an 
Expert-optimal $p^*$-equilibrium such that Client's possible priors are $\set{0, q_2}$.
Among these equilibria, as $q_2$ increases from   $\frac{c_2 - c_1}{l_2 - l_1}$ to $1$,
Client's equilibrium value increases continuously from
$0$ to $S(\prior) - \ev^*_χ (\prior)$.  

\subsection{Proof of \Cref{prp:welfare}} \label{app:prp-welfare}

Fixing some $\pvec \in P^2$,
Expert's $\pvec$-equilibrium value for this case
has been analyzed in \Cref{sec:p-equil}:
\[ 
v_\chi^* (\prior; \pvec) = 
\begin{cases}
  p_1 - c_1 & \text{ if } \prior < \qlbar(χ ; \pvec)  ; \\ 
  \cav v (\prior; \pvec) & \text{ if } \qlbar(χ ; \pvec) \leq \prior < \qbar(\pvec); \\
  p_2 - c_2 & \text{ if } \prior \geq \qbar(\pvec).
\end{cases}
\]
  When $\prior \ge \qbar(\pvec)$, Client always chooses $a_2$ in an Expert-optimal $\pvec$-equilibrium and
  $v_\chi^* (\prior; \pvec) = \prior l_2 + (1-\prior) l_1 - p_2$.
  When $\qlbar(χ ; \pvec) \le \prior < \qbar(\pvec)$, there exists a unique Expert-optimal $\pvec$-equilibrium
  outcome distribution, whose support is 
  $\set{(0, p_1 - c_1), (\qbar(\pvec), p_2 - c_2)}$.
  At posterior $0$ (resp.\ $\qbar(\pvec)$), Client always chooses 
  $a_1$ (resp.\ $a_2$).
  Therefore,
  $u_χ^*(\prior; \pvec) = (1 - \frac{\prior}{\qbar(\pvec)}) (l_1 - p_1) + 
  \frac{\prior}{\qbar(\pvec)} [\qbar(\pvec) l_2 + (1-\qbar(\pvec)) l_1 - p_2]$.
  When $\prior < \qlbar(χ ; \pvec)$, multiple Expert-optimal $\pvec$-equilibria exist.
  Among them, Client's equilibrium value is maximized 
  when the outcomes are 
  $(0, p_1 - c_1)$ and $(\qbar(\pvec), p_2 - c_2)$.
  As a result, Client's highest equilibrium value is still 
  $(1 - \frac{\prior}{\qbar(\pvec)}) (l_1 - p_1) + 
  \frac{\prior}{\qbar(\pvec)} [\qbar(\pvec) l_2 + (1-\qbar(\pvec)) l_1 - p_2]$.

Then, fix a price list $\pvec \in P^1$.
The concave and quasiconcave envelopes of $v (\prior ; \pvec)$ are:
\begin{equation}\label{eq:qcav-v-2}
  \qc v (\prior ; \pvec) = \begin{cases}
    p_1 - c_1 & \text{ if } \prior \in [0,\tilde q (\pvec)];  \\
    p_2 - c_2 & \text{ if } \prior \in (\tilde q (\pvec), 1],
  \end{cases}    
\end{equation}
\begin{equation}\label{eq:cav-v-2}
  \cav v (\prior ; \pvec) = \begin{cases}
    p_1 - c_1 & \text{ if } \prior \in [0,\tilde q (\pvec)];  \\
    \tilde \lambda(\pvec) (1 - \prior)  + p_2 - c_2 & \text{ if } \prior \in (\tilde q (\pvec), 1],
  \end{cases}    
\end{equation}
where the cutoff $\tilde q (\pvec)$ is
\begin{equation}
\tilde q (\pvec) = \begin{cases}
  \frac{l_1 - p_1}{l_1} & \text{if } p_2 \ge l_2 - \frac{l_2 - l_1}{l_1} p_1 \\
  \frac{p_2 - p_1}{l_2} & \text{otherwise.}
\end{cases}\label{eq:qtilde}
\end{equation}
and $\tilde \lambda(\pvec) = \frac{(p_1 - c_1) - (p_2 - c_2)}{1 - \tilde q (\pvec)}$.

Expert's $\pvec$-equilibrium value can be obtained
using the same procedures as in \Cref{sec:p-equil}:
\[ 
v_\chi^* (\prior; \pvec) = 
\begin{cases}
  p_1 - c_1 & \text{ if } \prior \le \tilde q(\pvec)  ; \\ 
  \cav v (\prior; \pvec) & \text{ if } \tilde q(\pvec) < \prior \le \hat q (\chi; \pvec); \\
  p_2 - c_2 & \text{ if } \prior > \hat q (\chi; \pvec).
\end{cases}
\]
where $\hat q  (\chi; \pvec) = \frac{\tilde q(\pvec)}{1 - (1-\tilde q(\pvec))\chi}$.
The characterization of $u_χ^*(\prior; \pvec)$ is obtained using the same argument for the previous case:
$$
u_χ^*(\prior; \pvec) = \begin{cases}
  (1 - \prior) l_1 - p_1 & \text{ if } \prior \in [0,\tilde q (\pvec)]; \\
  \frac{1 - \prior}{1 - \tilde q (\pvec)} \big ((1 - \tilde q (\pvec)) l_1 - p_1 \big) + \frac{\prior - \tilde q (\pvec)}{1 - \tilde q (\pvec)} (l_2 - p_2) & \text{ if } \prior \in (\tilde q (\pvec), 1].
\end{cases}
$$

Finally, fixing some $\pvec \in \Pbar$,
Client achieves the highest equilibrium value when Expert discloses all information to her. So
\[
u_χ^*(\prior; \pvec) = \prior (l_2 - p_2) + (1 - \prior) (l_1 - p_1).
\]

\subsection{Proof of \Cref{cor:discontinuity}}\label{app:cor-discontinuity}

Fix some $\pbar \in \Pbar$. Consider the limiting difference between 
$u_χ^*(\prior; \pbar)$ and $u_χ^* | _{P^1} (\prior; \pvec)$ as $\pvec$ converges to 
$\pbar$ from all possible directions in the region $P^1$.
By \Cref{prp:welfare},
$u_χ^* | _{P^1} (\prior ; \pvec)$ will converge to either
$(1 - \prior) l_1 - \bar{p}_1$ or  
$\frac{1 - \prior}{1 - \tilde{q}(\pvec) } \big ((1 - \tilde{q}(\pvec)) l_1 - p_1 \big) + \frac{\prior - \tilde{q}(\pvec) }{1 - \tilde{q}(\pvec) } (l_2 - p_2)$ as $\pvec  \to \pbar$.
In the first case, the limiting difference 
is $\prior (l_2 + p_1 - p_2) > 0$.
In the second case, let
$g (\qtilde) \equiv \frac{1 - \prior}{1 - \tilde q } \big ((1 - \tilde q ) l_1 - p_1 \big) + \frac{\prior - \tilde q}{1 - \tilde q } (l_2 - p_2)$
and then
$u_χ^*(\prior; \pbar) = g (0)$.
Since $g'(\qtilde) < 0$ for $\qtilde \in (0,1)$, the limiting difference
$u_χ^*(\prior; \pbar) -  g(\qtilde(\pbar))
= g(0) -  g(\qtilde(\pbar))$ is positive.

Similarly, 
$u_χ^* | _{P^2} (\prior ; \pvec)$ converges to either
$\prior l_2 + (1 - \prior) l_1 - \bar{p}_2$
or 
$(1 - \frac{\prior}{\qbar}) (l_1 - \bar{p}_1) + 
\frac{\prior}{\qbar} [\qbar l_2 + (1-\qbar) l_1 - \bar{p}_2]$ as $\pvec  \to \pbar$.
In the first case, the limiting difference is $(1 - \prior) (\bar{p}_2 - \bar{p}_1) > 0$. In the second case, let
$h (\qbar) \equiv (1 - \frac{\prior}{\qbar}) (l_1 - \bar{p}_1) + 
\frac{\prior}{\qbar} [\qbar l_2 + (1-\qbar) l_1 - \bar{p}_2]$ and then
$u_χ^*(\prior; \pbar) = h (1)$.
Since $h'(\qbar) > 0$ for $\qbar \in (0,1)$, the limiting difference
$u_χ^*(\prior; \pbar) - h(\qbar(\pbar)) = h(1) - h(\qbar(\pbar))$ is positive.

\subsection{Unknown problem type\label{app:unknown-type}}

This section contains discussions that are omitted from \Cref{sec:discussion-type}. 

Consider the \textit{alternative game} $\cG^A$ which differs from
$\cG$ only in that Expert cannot observe the problem directly.
Denote by $\cG^A_{\pvec}$ the subgame following Expert posting some price list $\pvec$ in $\cG^A$.
We refer to an equilibrium of
$\cG_{\pvec}^A$ as a \textit{$\pvec$-equilibrium of $\cG^A$}.
A $\pvec$-equilibrium of $\cG^A$ consists of four maps:
Expert's chosen experiment: $ξ^A: T \to Δ M$;
Expert's signalling strategy $σ^A: M \to Δ M$;
Client's strategy $ρ^A : M \to Δ A$;
and Client's belief updating
$η^A : M \to Δ T$; such that 
\begin{enumerate}[label=A\arabic*]
\item\label{itm:belief} 
  Given the experiment $ξ^A$
  and signalling strategy $σ^A$, Client's posterior after receiving message $m \in \M$ is
  obtained using Bayes's rule:
  $$
  \eta^A(\type \mid \m) = \frac
      {\Prior(t) [\chi ξ^A (m \mid t) + (1 - \chi) \sum\limits_{m' \in M}  σ^A (m \mid m') ξ^A (m' \mid t) ]}
      {\sum\limits_{t' \in T} \Prior (t') [\chi ξ^A(m \mid t') + (1 - \chi) \sum\limits_{m' \in M}  σ^A (m \mid m') ξ^A (m' \mid t)] } \, \text{ for $t \in T$;}
  $$

\item
  Given the belief updating $\eta^A$, Client's strategy
  $\rho^A (\wcard \mid \m)$ after receiving message $m$ is supported on 
  $$
  \arg\max_{a \in A} \sum_{\type \in \Type} 
    \eta ^A (\type \mid \m) u^C(a, \type , \pvec);
  $$ 
  
\item\label{itm:signalling} 
  Given Client's strategy $\rho^A$, Expert's signalling strategy
  $σ ^A (\wcard \mid m')$ is supported on 
  $$
  \arg\max_{m \in \M} \sum_{a \in A} \rho ^A (a \mid \m) u^E (a, \pvec)
  $$ 
  for all $m' \in \bigcup\limits_{t \in T} \supp ξ^A(\wcard \mid t)$.
\end{enumerate}

The $\pvec$-equilibrium of $\cG^A$ differs from that of $\cG$ (\Cref{def:p-eq})
in Client's belief updating (\ref{itm:belief}) and in Expert's choice of signalling strategy (\ref{itm:signalling}).

\begin{claim}\label{clm:alternative-game}
If $(ξ^A, σ^A, ρ^A, η^A)$ 
is a $\pvec$-equilibrium of $\cG^A$,
then $(ξ^A, \tilde{σ}, ρ^A, η^A)$
is a $\pvec$-equilibrium of $\cG$
where
$\tilde{σ} (m \mid t) = \sum_{m' \in M}  σ^A (m \mid m') ξ^A (m' \mid t)$. 
\end{claim}

\begin{proof}
Let $(ξ^A, σ^A, ρ^A, η^A)$ be a $\pvec$-equilibrium of $\cG^A$ and let $\widetilde{M} \equiv \bigcup_{t \in T} \supp ξ^A(\wcard \mid t)$.
Then $(ξ^A, \tilde{σ}, ρ^A, η^A)$ is a $\pvec$-equilibrium of $\cG$ as long as the following two conditions are satisfied.
\begin{enumerate}
\item\label{itm:ass1}
  For all $m \in \widetilde{M}$,
  Client's posterior after receiving $m$ using Bayes's rule is $η^A (t \mid m)$:
  $$
  \eta ^ A (\type \mid \m) = \frac
      {\Prior(t) [\chi ξ^A (m \mid t) + (1 - \chi) \tilde{σ} (m \mid t)]}
      {\sum_{t' \in T} \Prior (t') [\chi ξ^A (m \mid t') + (1 - \chi) \tilde{σ} (m \mid t') ] }\quad \text{ for all $t \in T$;}
  $$

\item\label{itm:ass2}
  For all $t \in T$,  
  Expert's signalling strategy $\tilde{σ} (\wcard \mid t)$ is supported on 
  $$
  \widetilde{M}_1 \equiv \arg\max_{m \in \widetilde{M}} \sum_{a \in A} \rho ^A (a \mid \m) u^E (a, \pvec).
  $$ 
\end{enumerate}

The validity of Condition \ref{itm:ass1} follows from the construction of $\tilde{σ}$.
For Condition \ref{itm:ass2},
let $A (m') = \supp {σ}^A (\wcard \mid m')$ and 
$B (t) = \supp {σ} (\wcard \mid t)$.
We want to show that given $σ^A (\wcard \mid m')$ is supported on
$\widetilde{M}_1$ for all $m' \in \widetilde{M}$,
$\tilde{σ} (\wcard \mid t)$ 
is also supported on 
$\widetilde{M}_1$ for all $t \in T$.
That is,
$$
A (m') \subseteq \widetilde{M}_1 \quad \forall m' \in \widetilde{M}
\implies 
B (t) \subseteq \widetilde{M}_1 \quad \forall t \in T.
$$
Or equivalently,
$$
\bigcup \limits_{m' \in \widetilde{M}} A (m') \subseteq \widetilde{M}_1
\implies 
\bigcup \limits_{t \in T} B (t) \subseteq \widetilde{M}_1.
$$
It suffices to show that
$\bigcup _{m' \in \widetilde{M}} A (m') \supseteq \bigcup _{t \in T} B (t)$. 
Or equivalently,
\begin{changemargin}{0cm}{-1.0cm}
\begin{align}
  & \qquad \qquad \qquad m \in \bigcup \limits_{t \in T} B (t) 
  \implies 
  m \in \bigcup \limits_{m' \in \widetilde{M}} A (m') \notag
 \\
 \iff & \quad m \in B (t) \text{ for some } t \in T 
  \implies 
 m \in A (m') \text{ for some } m' \in \widetilde{M}
 \notag
 \\ 
 \iff &
 \tilde{σ} (m \mid t) > 0 \text{ for some } t \in T 
 \implies
 σ^A (m \mid m') > 0 \text{ for some } m' \in \widetilde{M}  \label{eq:finally}
\end{align}
\end{changemargin}
The validity of \eqref{eq:finally} follows from the
construction that $\tilde{σ} (m \mid t) = \sum_{m' \in M}  σ^A (m \mid m') ξ^A (m' \mid t)$.
\end{proof}

\Cref{clm:alternative-game} guarantees that $\cG$ has ``more $\pvec$-equilibria'' compared to $\cG^A$.
So the rest of the discussions in \Cref{sec:discussion-type} follow.

\subsection{Public credibility \label{app:public-credibility}}

This section contains discussions that are omitted from \Cref{sec:discuss-public-credibility}.
Using the same argument for \Cref{lmm:1}, we can restrict attention to price lists 
$\pvec \in \Pbar \cup P^2$.
Below we show that it is always optimal for Expert to set $p_1^{\pc} = l_1$.

\begin{claim}
  With public credibility, it is optimal for Expert to set $p_1^{\pc} = l_1$ for any $\prior \in (0,1)$ and $χ \in (0,1)$.
\end{claim}

\begin{proof}
  Suppose there exists some optimal $\tilde{\pvec}= (\tilde p_1, \tilde p_2)$ where
  $\tilde p_1 < l_1$ and $\tilde p_1 - c_1 \le \tilde p_2 - c_2$.
  When $\tilde p_2 - c_2 \ge l_1 - c_1$,
  consider another price list $\pvec' =  (p_1', p_2')$
  with $p_1' = l_1$ and $p_2' = \tilde p_2$.
  It is straightforward to verify that
  $\qcav v (\prior ;\pvec') \ge \qcav v (\prior ; \tilde \pvec)$ and 
  $\cav v (\prior ;\pvec') \ge \cav v (\prior ; \tilde \pvec)$ for all $\prior \in (0,1)$.
  Therefore, $\pvec'$ is also optimal for Expert.

  When $\tilde p_2 - c_2 < l_1 - c_1$, then Expert's ex ante payoffs under $\tilde{\pvec}$ must be strictly less than $l_1 - c_1$. Again, Expert can obtain the payoff of $l_1 - c_1$ by setting $p_1 = l_1$ and $p_2 = l_1 - c_1 + c_2$ and disclosing all information to Client.  
  Such $\tilde{\pvec}$ cannot be Expert-optimal. 
\end{proof}

Given $p_1^{\pc} = l_1$, we have $\qbar (\pvec) = \frac{p_2 - l_1}{l_2 - l_1}$.
Then Expert's $\pvec$-equilibrium value 
$v_{χ}^{\pc}(\prior; \pvec) = χ \cav v (\prior ; \pvec) + (1 - χ) \qc v (\prior ; \pvec)$
can be written as:
\begin{equation} \label{eq:public-eq-value}
  v_{χ}^{\pc}(\prior; l_1, p_2) =
  \begin{cases}
     \chi \prior l_2 + (1 - \chi \prior) l_1 - c_1 - \frac{\chi \prior (l_2 - l_1) (c_2 - c_1)}{p_2 - l_1}  & \textif 
     \prior < \frac{p_2 - l_1}{l_2 - l_1} 
     \\
     p_2 - c_2 & \textif 
     \prior \geq \frac{p_2 - l_1}{l_2 - l_1} 
  \end{cases}
\end{equation}
Below we show that the optimal $p_2$ is either
$l_2$ or $\prior l_2 + (1-\prior) l_1$.

\begin{claim}\label{clm:public-optimal-p2}
  Given $p_1^{\pc} = l_1$, the optimal $p_2$ for Expert is either $l_2$ or 
  $\prior l_2 + (1-\prior) l_1$.      
\end{claim}

\begin{proof}
Fix some $\prior \in (0,1)$ and consider the maximal value of
$v_{χ}^{\pc}(\prior; l_1, p_2)$ over $p_2 \in [l_1 -c_1 + c_2, l_2]$.
Since the condition $\prior < \frac{p_2 - l_1}{l_2 - l_1}$ can be rewritten equivalently as
$p_2 > \prior l_2 + (1-\prior) l_1$,
\Cref{eq:public-eq-value} implies that
$v_{χ}^{\pc} (\prior; l_1, p_2) < 
v_{χ}^{\pc} (\prior; l_1, l_2)$
for all $p_2 > \prior l_2 + (1-\prior) l_1$.
Similarly, 
$v_{χ}^{\pc} (\prior; l_1, p_2) < 
v_{χ}^{\pc} (\prior; l_1, \prior l_2 + (1-\prior) l_1)$
for all $p_2 < \prior l_2 + (1-\prior) l_1$.
To sum up,   $v_{χ}^{\pc}(\prior; l_1, p_2)$
achieves its maximum at either
$p_2=l_2$ or $p_2=\prior l_2 + (1-\prior) l_1$.
\end{proof}

As per \Cref{clm:public-optimal-p2},
we can determine Expert's equilibrium value as the upper envelope of $v_{χ}^{\pc}(\prior; l_1, l_2)$
and $v_{χ}^{\pc}(\prior; l_1, \prior l_2 + (1-\prior) l_1)$, as depicted in \Cref{fig:modification}.
The corresponding optimal $p_2$ is: 
$$
p_2^{\pc} (\prior, χ) = \begin{cases}
  l_2 
    & \textif  \prior < q_\chi^{\pc}\\
  l_2 \text{ or } \prior l_2 + (1-\prior) l_1 
    & \textif  \prior = q_\chi^{\pc}\\
  \prior l_2 + (1-\prior) l_1
    & \textif  \prior > q_\chi^{\pc}
\end{cases}  
$$
where 
$q_\chi^{\pc} = \frac{c_2 - c_1}{(1 - \chi) (l_2 - l_1) + \chi (c_2 - c_1)}$.

\Cref{clm:public-optimal-p2} 
implies that there are essentially two kinds of equilibria.
First, Expert sets $p_1 = l_1$ and $p_2 = l_2$, 
profiting either $\qcav v(\prior ; \pvec)$ or
$\cav v(\prior ; \pvec)$ depending on his realized credibility.
Second, Expert discloses no information to Client and sets the monopoly price 
$p_2 = \prior l_2 + (1- \prior) l_1$.
In both kinds of equilibria, Client does not benefit from Expert's services.

\end{document}